\documentclass{article}

\usepackage{graphicx} 
\usepackage{xcolor}

 \usepackage[textwidth=158mm, hcentering, vmargin={24mm,32.5mm}, foot=20mm]{geometry}

\usepackage{booktabs} 
\usepackage{float}

\title{Encoder--Decoder Transformers: \\Logical Characterizations and Periodicity}

\usepackage{authblk}

\author[1]{Veeti~Ahvonen}
\author[1]{Damian~Heiman}
\author[1]{Antti~Kuusisto}
\author[1]{Miguel~Moreno}
\author[1]{Matias~Selin}
\affil[1]{Mathematics Research Centre, Tampere University, Tampere, Finland}

\date{September 2026}

\usepackage{tikz}
\usepackage{multicol}
\usepackage{graphics}

\usepackage{arydshln}
\usepackage{makecell}

\usepackage{colonequals}
\usepackage{amsmath, amssymb, amsthm}
\usepackage{mathtools}
\usepackage{thm-restate}

\usepackage{url}
\usepackage{hyperref}

\usepackage{marginnote}

\usepackage{graphicx} 

\usepackage{verbatim}

\usepackage[utf8]{inputenc}
\usepackage[T1]{fontenc}

\usepackage{enumitem}
\usepackage{tikz}
\usepackage{lipsum}
\usepackage{pdfpages}
\usepackage{multicol} 
\usepackage{xspace} 
\usepackage{stmaryrd}

\usetikzlibrary{positioning, arrows.meta, calc, fit, backgrounds, decorations.pathreplacing}
\PassOptionsToPackage{numbers, compress}{natbib}

\usepackage{makecell}
\usepackage{mathtools}
\usepackage{colonequals}

\usepackage{marginnote}

\usepackage{bm}

\theoremstyle{plain}
\newtheorem{theorem}{Theorem}
\newtheorem{lemma}[theorem]{Lemma}
\newtheorem{proposition}[theorem]{Proposition}

\theoremstyle{definition}

\newtheorem{definition}[theorem]{Definition}

\newcommand{\N}{\mathbb N}
\newcommand{\Z}{\mathbb Z}

\newcommand{\bb}{\mathbf{b}}

\newcommand{\be}{\mathbf{e}}
\newcommand{\bbf}{\mathbf{f}}

\newcommand{\bs}{\mathbf{s}}

\newcommand{\bu}{\mathbf{u}}
\newcommand{\bv}{\mathbf{v}}

\newcommand{\bx}{\mathbf{x}}

\newcommand{\cF}{\mathcal{F}}

\newcommand{\cM}{\mathcal{M}}

\newcommand{\cV}{\mathcal{V}}
\newcommand{\cW}{\mathcal{W}}

\newcommand{\softmax}{\mathrm{softmax}}
\newcommand{\T}{\mathsf{T}}
\newcommand{\ReLU}{\mathrm{ReLU}}

\newcommand{\SUM}{\mathrm{SUM}}

\newcommand{\MLP}{\mathrm{MLP}}
\newcommand{\md}{\mathsf{md}}
\newcommand{\width}{\mathsf{width}}

\newcommand{\BOS}{\mathrm{BOS}}
\newcommand{\EOS}{\mathrm{EOS}}

\newcommand{\concat}{\mathrm{concat}}
\newcommand{\GPTL}{\mathrm{GPTL}}
\newcommand{\GPTLF}{\mathrm{GPTL}^{-}}
\newcommand{\CPG}{\mathrm{CPG}}
\newcommand{\EL}{\mathsf{EL}}
\newcommand{\DL}{\mathsf{DL}}
\newcommand{\MHSA}{\mathsf{MHSA}}
\newcommand{\MHCA}{\mathsf{MHCA}}

\begin{document}

\maketitle

\begin{abstract}
  \noindent
  We give logical characterizations of encoder--decoder transformers, the foundational architecture for LLMs that also sees use in various settings that benefit from cross-attention, in the practical setting of floating-point numbers and soft attention. First, we characterize such transformers via a new temporal logic that extends propositional logic with a counting global modality over the encoder input and a past modality over the decoder input, as well as via a type of distributed automata. We consider three frameworks: with and without a final softmax step in the transformer, and in the setting where each model generates tokens via autoregression. Second, we show that both autoregressive transformers and sentences of counting propositional logic---the fragment of the previous logic obtained by omitting the past modality---recognize exactly the commutative star-free languages. Finally, we find that the sequences of tokens the transformers generate are ultimately periodic (and each token appears in the period at most once). This allows us to characterize autoregressive transformers via sentences of counting propositional logic that generate tokens without autoregression, i.e., we can effectively eliminate recursion from the transformers.
\end{abstract}

\section{Introduction}

Transformers \cite{vaswani2017attention} are the dominant architecture in modern 
natural language processing. 
They
take words as inputs and incorporate various attention mechanisms and neural networks to produce numeric representations of the meanings of words (or more precisely, components of words called tokens). These representations are iteratively refined and used to sample from a selection of tokens to generate text. Due to the ubiquity of transformers, their expressive power has become an active area of 
research, with a rich body of results connecting to formal language 
theory, temporal logic, and circuit complexity \cite{Strobl2024}.

So far, nearly all such results have concerned either encoder-only transformers (using unmasked self-attention) or decoder-only transformers (using masked self-attention). The original encoder--decoder architecture of \cite{vaswani2017attention} that combines unmasked and 
masked self-attention with cross-attention has been mostly absent, with the sole exception of the Turing-completeness construction of \cite{perez2021attention}; notably, no fine-grained characterization exists. This gap is significant as encoder--decoder models remain standard for various settings, such as machine translation \cite{nllb2024scaling} and connecting distinct modalities like text, audio, and video via cross-attention \cite{zheng2024open, sekoyan2025canary}.

In this paper, we fill this gap with a detailed investigation into the formal expressivity of floating-point soft-attention encoder--decoder transformers.\footnote{We consider transformers without positional encodings, the study of which we leave for future work. Our results hold for some variations in the transformer architecture regarding whether attention layers have multiple heads, the strictness of the masking and the inclusion of layer normalization; see Appendix \ref{appendix:transformer-variations}.} We first give novel logical characterizations at multiple stages of transformer output via a new temporal logic and a class of distributed automata. While these results are of independent interest, they are also necessary stepping stones toward a logical characterization of autoregressive encoder--decoder transformers in the language recognition setting. They also lead to an interesting observation about the periodicity of floating-point transformers and a subsequent characterization of autoregressive transformers---in terms of generated token sequences---that essentially eliminates autoregression.

In Section~\ref{section: skeleton section}, we characterize transformers that simply output feature vectors, i.e., vectors of floats; this is a fundamental building block toward later results. In particular, we show that in this setting encoder--decoder transformers are characterized by the logic $\GPTLF$ 
which extends propositional logic with a counting global modality over the encoder input and a past modality over the decoder input. We also show that they are expressively equivalent to a class of distributed automata called counting past-global distributed automata ($\CPG$-automata) where the state transition function depends on a vertex's own previous state and two bounded multisets of previous states, one attending to vertices in the encoder input and the other attending to prior vertices in the decoder input.

\begin{restatable*}{theorem}{skeletontheorem}\label{theorem:main}
Encoder--decoder transformers without the final $\softmax$, $\GPTLF$ and $\CPG$-automata have the same expressive power.
\end{restatable*}

In practice, transformers output a probability distribution over the available tokens via the $\softmax$ function.
We give a characterization in this setting in Section~\ref{section: similarity section}. Since $\softmax$ cannot output every bit string, 
we use a similarity relation $\sim$ which intuitively connects nearby output vectors together; objects are equivalent if they produce similar outputs whenever possible.

\begin{restatable*}{theorem}{similaritytheorem}\label{theorem: softmax characterization}
    For each similarity relation $\sim$, encoder--decoder transformers with the final $\softmax$, $\GPTLF$ and $\CPG$-automata have the same expressive power w.r.t. $\sim$.
\end{restatable*}

The purpose of the probability distribution is to enable autoregressive token generation; we study this setting in Section \ref{section:inference-loop}. Here, transformers, $\GPTL^-$ and $\CPG$-automata work in a periodic fashion to generate an output string from an input string. The output string is initialised as a beginning-of-string token. Iteratively, the input and the so-far generated output are fed to the transformer (respectively, tuple of $\GPTL^-$ or CPG-automaton), after which a token is sampled based on the output using a function that maps bit strings to tokens. The sampled token is appended to the end of the output, and the process repeats until an end-of-string token is produced. Two objects are generation-equivalent if they generate the same outputs with the same inputs.

\begin{restatable*}{theorem}{generationtheorem}\label{theorem: generation-equivalence}
    As autoregressive token-generating systems, encoder--decoder transformers, $\GPTLF$ and $\CPG$-automata are equally expressive regardless of if the transformers use the final $\softmax$ or not.
\end{restatable*}

With these characterizations in hand, we study autoregressive transformers further. In Section~\ref{section: characterizing accepting transformers via CPL}, we first investigate 
the setting of language recognition; a word is fed to the encoder and is accepted iff an accepting token is generated at any point. We show that the commutative star-free languages are exactly those recognized by transformers, counting propositional logic CPL (the fragment of $\GPTL^-$ without the past modality) and $k$-capped multiset functions, which accept inputs based on how many (up to the bound $k$) of each token they have.

\begin{restatable*}{theorem}{fourwayresult}\label{theorem: language characterization}
    Sentences of CPL, accepting encoder--decoder transformers and $k$-capped multiset functions recognize exactly the commutative star-free languages.
\end{restatable*}

A natural follow-up is to ask how much of the autoregression we can capture via CPL.
Theorem~\ref{theorem: generation-equivalence} reveals a useful insight into answering this question: transformers are ultimately periodic, meaning that if the generated output sequence is infinite, it eventually begins looping indefinitely, with each loop containing at most one of each token. 

\begin{restatable*}{proposition}{ultimatelyperiodic}\label{proposition: ultimately periodic}
    Each encoder--decoder transformer is ultimately periodic and the periods are repetition-free.
\end{restatable*}

The culmination of this work in Section~\ref{section: characterizing generating transformers via CPL} involves using this insight to characterize the \emph{entire token sequences} generated by autoregressive transformers in terms of sentences of CPL arranged into two matrices, the first corresponding to the prefix of the sequence and the second to the repeating period.

\begin{restatable*}{theorem}{matrixtheorem}\label{theorem: transformer cpl generation equivalence}
    As token-generating systems, autoregressive encoder--decoder transformers and sentences of CPL (without autoregression) are equally expressive regardless of if the transformers use the final $\softmax$ or not.
\end{restatable*}

We conclude the work and discuss limitations in Section~\ref{section: conclusion}. Our results constitute a major step forward in the formal study of encoder--decoder transformers. Our result that transformers are ultimately periodic also naturally raises many questions regarding the more general relationship between autoregression and expressivity, which we leave for future work.

\paragraph{Related work.}
The expressive power of fixed-precision transformers has received much attention in recent years. We focus here on the work directly relevant to this paper; for more, see Appendix \ref{appendix: related work}.

The papers closest to our work are \cite{licharacterizing}, \cite{yang2026probability} and \cite{ahvonen2026}.
We first give an overview of their results and then contrast them with our work.
Li \& Cotterell \cite{licharacterizing} showed that floating-point decoder-only transformers recognize exactly the left-deterministic polynomial languages and characterized such transformers via the past-only fragment of linear temporal logic
(LTL[P]) 
as well as via partially ordered deterministic finite automata.
Yang et al. \cite{yang2026probability} extended this characterization (as well as corresponding ones for real-valued transformers and unique hard-attention) to the language modeling setting where equivalence is over probability distributions. 
Ahvonen et al. \cite{ahvonen2026} characterized floating-point encoder-only transformers via propositional logic extended with a counting global modality; this also applies to transformers over graphs.

We emphasize our contributions in contrast to the above papers.
First, while the mentioned papers concern encoder-only and decoder-only transformers, we characterize encoder--decoder transformers which combine the two via cross-attention in a non-trivial way. Our characterizations are via two different logics and a new type of distributed automata, and they reveal that encoder--decoder transformers recognize exactly the commutative star-free languages.
Second, our language recognition framework works natively via token generation unlike the common approach in the literature adopted in \cite{licharacterizing} where transformers are treated as binary classifiers directly; we also show that these two frameworks are equivalent. Finally, our work also goes in an entirely new direction: we observe that the token sequences produced by floating-point transformers are ultimately periodic and the resulting periods contain at most one of each token.
We use this to give a whole new type of result where an autoregressive transformer is characterized via a logic that is not used autoregressively.

Regarding other characterizations of decoder-only transformers, Yang et al. \cite{yang2024masked} showed that over reals masked rightmost-hard-attention transformers recognize exactly the star-free languages, establishing an equivalence with linear temporal logic. 
For fixed precision,
Jerad et al. \cite{jerad2025} used LTL[P] to capture leftmost-hard attention transformers. 
Yang et al. \cite{yang2025} characterized transformers that round to fixed precision except inside attention via LTL[P] extended with counting operators.
We show that soft-attention floating-point encoder--decoder transformers recognize exactly the commutative star-free languages, establishing them as a less expressive variant than the transformers in \cite{yang2024masked}, \cite{licharacterizing}, \cite{jerad2025} and~\cite{yang2025}.

To our knowledge, the sole previous work on the formal expressivity of encoder--decoder transformers is the work of Pérez et al. \cite{perez2021attention}. They consider autoregressive transformers in terms of language recognition and show that they are Turing-complete with rational numbers of arbitrary precision, average-hard attention and positional encodings. They also note that 
without positional encodings or with fixed precision such transformers cannot separate words with the same ratios of tokens. 
By contrast, most of our results are more general than language recognition,
and in the case of language recognition we give a precise characterization of floating-point encoder--decoder transformers with soft attention and no positional encodings via the commutative star-free languages. 
This reveals orthogonal expressivity between arbitrary precision and floats in the framework of encoder--decoder transformers and no positional encodings (not taking into account the subtle differences between soft and average hard-attention): with floats they cannot recognize certain properties recognizable with arbitrary precision such as ``there are more tokens $a$ than tokens $b$'', but they can recognize some properties not recognizable with arbitrary precision such as ``there are precisely five tokens $a$''. This is consistent with prior observations by Ahvonen et al. \cite{ahvonen2026}.

\section{Preliminaries}\label{section: preliminaries}

Let $\N = \{0, 1, \dots\}$ and $\Z_+ = \{1, 2, \dots\}$.
For all $n \in \Z_+$, 
let $[n] := \{1, \dots, n\}$.
We denote finite sets of \textbf{tokens} by $\Sigma$. 
We often denote strings and sequences by bold letters $\bs$ or arrows $\vec s$.
A \textbf{multiset over} $S$ is a function $S\to\N$, and a $k$\textbf{-capped multiset over} $S$ is a function $S\to\{0,\dots,k\}$. We denote the set of multisets over $S$ as $\mathcal{M}(S)$ and the set of $k$-capped multisets over $S$ as $\mathcal{M}_k(S)$.

\subsection{Word-shaped graphs}

Transformers operate on words which we represent as graphs.
A \textbf{$\Sigma$-labeled word-shaped graph} $\mathcal{W}$ is a tuple $(V, E, \lambda)$ where $V = [n]$ ($n \in \Z_+$) is a set of \textbf{vertices}, $E$ is the successor relation $\{(i, i+1) \,|\, i \in [n-1]\}$ (inducing a strict linear order $<$ over $V$), and $\lambda : V \to \Sigma$ is a \textbf{labelling function} assigning a token to each vertex. We may use $V(\cW)$ to denote $V$. A \textbf{pointed $\Sigma$-labeled word-shaped graph} is a pair $(\cW,i)$ where $i \in V(\cW)$. 

In particular, a transformer operates on two words, one fed to the encoder and the other to the decoder. We represent the two inputs as a two-sorted graph.
A \textbf{$\Sigma$-labeled two-sorted word-shaped graph} 
is a $\Sigma \cup \{\BOS\}$-labeled graph $\cW$ where exactly one vertex~$i$ has the token $\BOS$ (``Beginning Of String'').
The \textbf{prefix} and the \textbf{suffix} are the subgraphs induced by vertices $1, \dots, i-1$ and $i, \dots, n$ respectively.
In the autoregressive setting, we will also call the prefix the \textbf{input} and the suffix the \textbf{output}.
For simplicity, we will refer to word-shaped graphs as words. We also refer to token strings as words, but the meaning is always clear from the context.

\subsection{Floating-point numbers}

We study encoder--decoder transformers with fixed numerical precision. Each transformer uses a single floating-point format for all calculations, but different transformers may use different formats.

Given positive integers $p,q\in\mathbb{Z}_+$, a \textbf{floating-point number} (or \textbf{float}) over $p$ and $q$ is a bit string
\[
  b_0 \cdots b_{q+p} \in \{0,1\}^{p+q+1}.
\]
The bit $b_0$ encodes the \textbf{sign}, $\be \colonequals b_1\cdots b_q$ is the \textbf{exponent}, representing $e \in \N$ in binary, and $b_{q+1}\cdots b_{q+p}$ is the \textbf{significand}, representing $s \in \N$ in binary.
Let $a=2^{p-1}$ and $b=2^{q-1}-1$.  
Except for special cases described below, the float $b_0\cdots b_{q+p}$ encodes the real number
$
  (-1)^{b_0}\frac{s}{a}2^{\,e-b}.
$
If $\mathbf{s}=0^p$ and $\mathbf{e}=1^q$, the float represents $\infty$ for $b_0=0$ and $-\infty$ for $b_0=1$.
When unambiguous, we identify a float with the real value (or $\pm\infty$) it represents. 

A float is called \textbf{normalized} if $b_{q+1}=1$, and \textbf{subnormalized} if $b_{q+1}=0$ and $\mathbf{e}=0^q$.
The \textbf{floating-point format} $\mathcal{F}(p,q)$ consists of all normalized and subnormalized floats over $p$ and $q$, together with $\infty$, $-\infty$, and a designated float for NaN (Not a Number). When $p$ and $q$ are clear from the context, we simply write $\cF$.

Arithmetic operations $+$, $-$, $\cdot$, $\div$, and $\sqrt{\phantom{x}}$ over $\cF(p,q)$ are defined as follows. First compute the exact result $r$ in extended real arithmetic, and then map it to a float in $\cF(p,q)$ as follows. If $r$ is undefined (e.g.,\ $\infty/\infty$) or an operand is NaN, the result is NaN. If $r$ is representable in $\cF(p,q+1)$ but not in $\cF(p,q)$, then $r$ exceeds the range of $\cF(p,q)$ and the result is $\pm\infty$ according to its sign. Otherwise, $r$ is rounded to the nearest float in $\cF(p,q)$. Special functions such as $\texttt{exp}(x)$ are approximated within the float format, e.g.,\ using Taylor or related series.

\subsection{Encoder--decoder transformers}\label{subsection:encoder-decoder-transformers}

Our encoder--decoder transformer is modeled after the original architecture of Vaswani et al. \cite{vaswani2017attention}; see Figure~\ref{fig:vanilla-transformer}. In this section, we give an intuitive overview of this architecture; for detailed definitions, see Appendix \ref{appendix: encoderdecoder}, and for variations on masking, multi-head attention and layer normalization, see Appendix~\ref{appendix:transformer-variations}.

\begin{figure}[ht]
  \centering
\resizebox{390pt}{!}{
\begin{tikzpicture}[
    remember picture,
]

\definecolor{colorAttn}{RGB}{248,203,153}
\definecolor{colorAddNorm}{RGB}{255,230,153}
\definecolor{colorFF}{RGB}{162,196,228}
\definecolor{colorLinear}{RGB}{169,208,174}
\definecolor{colorSoftmax}{RGB}{192,170,212}
\definecolor{colorEncBlock}{RGB}{255,235,205}
\definecolor{colorDecBlock}{RGB}{220,230,245}

\tikzset{
    block/.style={
        draw, rectangle, minimum width=2.8cm, minimum height=0.65cm,
        font=\small, align=center, rounded corners=2pt
    },
    embedblock/.style={block, fill=black!10},
    attnblock/.style={block, fill=colorAttn},
    addnormblock/.style={block, fill=colorAddNorm},
    ffblock/.style={block, fill=colorFF},
    linearblock/.style={block, fill=colorLinear},
    softmaxblock/.style={block, fill=colorSoftmax},
    stackblock/.style={
        draw, thick, rectangle, minimum width=2.4cm, minimum height=0.8cm,
        font=\small, align=center, rounded corners=3pt
    },
    encstack/.style={stackblock, fill=colorEncBlock},
    decstack/.style={stackblock, fill=colorDecBlock},
    arrow/.style={-{Stealth[length=5pt]}, thick},
    lbl/.style={font=\small},
    brace/.style={
        decorate, decoration={brace, amplitude=8pt}, thick, gray
    },
    bracemirror/.style={
        decorate, decoration={brace, amplitude=8pt, mirror}, thick, gray
    },
}


\node[lbl] (enc-input) {Prefix $\mathcal{W}_p$};
\node[embedblock, above=0.4cm of enc-input, minimum width=2.4cm] (enc-embed) {Prefix embedding};

\node[encstack, above=0.5cm of enc-embed] (enc-L1) {Encoder layer $\EL_1$};
\node[above=0.50cm of enc-L1, font=\small] (enc-dots) {$\vdots$};
\node[encstack, above=0.50cm of enc-dots] (enc-LN) {Encoder layer $\EL_{L_1}$};

\node[lbl, right=3.5cm of enc-input] (dec-input) {Suffix $\mathcal{W}_s$};
\node[embedblock, above=0.4cm of dec-input, minimum width=2.4cm] (dec-embed) {Suffix embedding};

\node[decstack, above=0.5cm of dec-embed] (dec-L1) {Decoder layer $\DL_1$};
\node[above=0.50cm of dec-L1, font=\small] (dec-dots) {$\vdots$};
\node[decstack, above=0.50cm of dec-dots] (dec-LN) {Decoder layer $\DL_{L_2}$};

\node[linearblock, above=0.5cm of dec-LN, minimum width=2.4cm] (linear) {Linear};
\node[softmaxblock, above=0.3cm of linear, minimum width=2.4cm] (softmax) {Softmax};
\node[lbl, above=0.4cm of softmax] (output) {\begin{tabular}{c}Output\\probabilities\end{tabular}};

\draw[arrow] (enc-input) -- (enc-embed);
\draw[arrow] (enc-embed) -- (enc-L1);
\draw[arrow] (enc-L1) -- (enc-dots);
\draw[arrow] (enc-dots) -- (enc-LN);

\draw[arrow] (dec-input) -- (dec-embed);
\draw[arrow] (dec-embed) -- (dec-L1);
\draw[arrow] (dec-L1) -- (dec-dots);
\draw[arrow] (dec-dots) -- (dec-LN);
\draw[arrow] (dec-LN) -- (linear);
\draw[arrow] (linear) -- (softmax);
\draw[arrow] (softmax) -- (output);

\draw[arrow] (enc-LN.north) -- ([yshift=20pt]enc-LN.north) -- ([yshift=20pt, xshift=70]enc-LN.north)
-- ([xshift=-35.5pt]dec-L1.west)
-- (dec-L1.west);
\draw[arrow] ([xshift=30pt]enc-LN.east) -- (dec-LN.west);
\draw[thick, dotted] ([xshift=32pt, yshift=-35pt]enc-LN.east) -- ([yshift=-34.5pt]dec-LN.west);


\node[attnblock, left=3.0cm of enc-L1, yshift=0cm] (edet-mha) 
    {Multi-head\\self-attention};
\node[addnormblock, above=0.3cm of edet-mha] (edet-an1) 
    {Add (\& Norm)};
\node[ffblock, above=0.3cm of edet-an1] (edet-ff) {MLP};
\node[addnormblock, above=0.3cm of edet-ff] (edet-an2) 
    {Add (\& Norm)};

\begin{scope}[on background layer]
\node[draw, thick, rounded corners=3pt, fill=colorEncBlock,
      fit=(edet-mha)(edet-an1)(edet-ff)(edet-an2),
      inner xsep=10pt, inner ysep=8pt] (edet-box) {};
\end{scope}

\draw[arrow] (edet-mha) -- (edet-an1);
\draw[arrow] (edet-an1) -- (edet-ff);
\draw[arrow] (edet-ff) -- (edet-an2);

\coordinate (edet-res-in) at ([yshift=-17pt]edet-mha.south);
\draw[arrow] ([yshift=-10pt]edet-res-in) -- (edet-mha);
\draw[arrow] (edet-res-in) -| ([xshift=-0.7cm, yshift=-5pt]edet-an1.west) 
    -- ([yshift=-5pt]edet-an1.west);
\draw[arrow] ([yshift=3pt]edet-an1.west) -| ([xshift=-0.7cm]edet-an2.west) 
    -- (edet-an2.west);

\coordinate (edet-res-out) at (edet-an2.north);
\draw[arrow] (edet-an2) -- ([yshift=23pt]edet-res-out);

\draw[bracemirror] 
    ([xshift=65pt, yshift=9pt]edet-res-in) -- ([xshift=65pt, yshift=8pt]edet-res-out)
    coordinate[midway] (enc-brace-mid);
\draw[thick, gray] (enc-L1.west) -- ([xshift=8.3pt]enc-brace-mid);


\node[attnblock, right=3.0cm of dec-L1, yshift=-0.8cm] (ddet-mmha) 
    {Masked multi-head\\self-attention};
\node[addnormblock, above=0.3cm of ddet-mmha] (ddet-an1) 
    {Add (\& Norm)};
\node[attnblock, above=0.3cm of ddet-an1] (ddet-mha) 
    {Multi-head\\cross-attention};
\node[addnormblock, above=0.3cm of ddet-mha] (ddet-an2) 
    {Add (\& Norm)};
\node[ffblock, above=0.3cm of ddet-an2] (ddet-ff) {MLP};
\node[addnormblock, above=0.3cm of ddet-ff] (ddet-an3) 
    {Add (\& Norm)};

\begin{scope}[on background layer]
\node[draw, thick, rounded corners=3pt, fill=colorDecBlock,
      fit=(ddet-mmha)(ddet-an1)(ddet-mha)(ddet-an2)(ddet-ff)(ddet-an3),
      inner xsep=10pt, inner ysep=8pt] (ddet-box) {};
\end{scope}

\draw[arrow] (ddet-mmha) -- (ddet-an1);
\draw[arrow] (ddet-an1) -- (ddet-mha);
\draw[arrow] (ddet-mha) -- (ddet-an2);
\draw[arrow] (ddet-an2) -- (ddet-ff);
\draw[arrow] (ddet-ff) -- (ddet-an3);

\coordinate (ddet-res-in) at ([yshift=-17pt]ddet-mmha.south);
\draw[arrow] ([yshift=-10pt]ddet-res-in) -- (ddet-mmha);
\draw[arrow] (ddet-res-in) -| ([xshift=0.7cm, yshift=-5pt]ddet-an1.east) 
    -- ([yshift=-5pt]ddet-an1.east);
\draw[arrow] ([yshift=3pt]ddet-an1.east) -| ([xshift=0.7cm, yshift=-5pt]ddet-an2.east) 
    -- ([yshift=-5pt]ddet-an2.east);
\draw[arrow] ([yshift=3pt]ddet-an2.east) -| ([xshift=0.7cm]ddet-an3.east) 
    -- (ddet-an3.east);

\draw[arrow] (ddet-res-in) -- (ddet-mmha);
\coordinate (ddet-res-out) at (ddet-an3.north);
\draw[arrow] (ddet-an3) -- ([yshift=23pt]ddet-res-out);

\draw[arrow] ([xshift=-0.6cm]ddet-mha.west) -- (ddet-mha.west);

\draw[brace] 
    ([xshift=-65pt,yshift=9pt]ddet-res-in) -- ([xshift=-65pt,yshift=8pt]ddet-res-out)
    coordinate[midway] (dec-brace-mid);
\draw[thick, gray] (dec-L1.east) -- ([xshift=-8.5pt]dec-brace-mid);

\end{tikzpicture}}
  \caption{The encoder--decoder transformer of Vaswani et al. \cite{vaswani2017attention} with positional encodings omitted. Layer normalization (``Norm'') is included for reference but does not affect expressivity; see Appendix~\ref{appendix:layer-normalization}.}
  \label{fig:vanilla-transformer}
\end{figure}

Our framework is floating-point numbers, soft attention, strict causal masking and no positional encodings. An \textbf{encoder--decoder transformer} with dimension $d$ has two parts: an encoder operating on prefixes $\cW_p$ and a decoder operating on suffixes $\cW_s$. First an \textbf{embedding function} $\Sigma \to \cF^d$ is applied vertex-wise. The encoder then applies a stack of \textbf{encoder layers}, each consisting of a self-attention sub-layer followed by an $\MLP$. The decoder applies a stack of \textbf{decoder layers}, each consisting of a masked self-attention sub-layer, a cross-attention sub-layer and an $\MLP$. Each sub-layer and $\MLP$ has input and output dimension $d$ and is wrapped in a \textbf{residual connection} that adds its input to its output. $\MLP$s have two layers and use $\ReLU(x) = \max\{0,x\}$ in the first.

All attention sub-layers are built from \textbf{attention heads}
\[
    H(X,Y) := \softmax\Big(
    \frac{(YW^Q)(XW^K)^{\top}}{\sqrt{d_k}}\Big)(XW^V),
\]
where $W^Q$, $W^K$, $W^V$ are weight matrices, $d_k$ is a hidden dimension and $\softmax$ is applied row-wise. A \textbf{self-attention head} is $H(X) := H(X,X)$. A \textbf{cross-attention head} $H(X,Y)$ combines the output $X$ of the encoder with the output $Y$ of the previous sub-layer of the decoder, so each decoder position attends to the whole encoder output. In \textbf{masked} self-attention, a \textbf{strict causal mask} sets the entries on and above the diagonal 
to $-\infty$ before $\softmax$, so positions attend only to strictly prior positions.
Attention sub-layers may apply many attention heads of the same kind in parallel; the outputs are concatenated followed by a linear mapping to return to dimension $d$. 

The function softmax is a length-preserving function on float vectors: for $\bx = (x_1, \dots, x_n)$ the $i$th component of $\softmax(\bx)$ is defined as
\[
    \frac{e^{x_i - \max(\bx)}}{\sum_{j = 1}^n e^{x_j - \max(\bx)}}.
\]
The denominator contains a sum of floats, but the sum of floats is not associative.
For example, in the format $\cF(3,3)$ we have $(1 + 4) + 5 = 5 + 5 = 10$ whereas $1 + (4 + 5) = 1 + 8 = 8$ because $9$ is not representable in $\cF(3,3)$ and rounds to $8$. Thus, we need to fix an order for summation.
Summing in the order of the vertices would allow the output of $\softmax$ in a particular position to depend on how the values in the other positions are ordered, which is not desirable.
We thus instead adopt the convention to always sum in the increasing order of floats, which is a reasonable choice w.r.t.\ accuracy \cite{higham_article, floats_robertazzi, wilkinson}.
This leads to the sum of a multiset of floats being bounded in the sense that it only depends on the multiplicities of floats up to some bound $k$. Formally, for a multiset $M$ of floats, let $\SUM_\cF(M)$ denote the value of the sum in which each float
$f$ appears $M(f)$ times and the summands are added in increasing order, and let $M_{|k}$ denote the \textbf{$k$-projection} of~$M$, i.e., $M_{|k}(x):=\text{min}\{M(x),k\}$. Then we obtain:

\begin{proposition}[\cite{ahvonen_2024}]\label{proposition: floating-point saturation}
    For each floating-point format $\cF$, there exists a $k \in \N$ such that for all multisets $M$ over floats in $\cF$, we have $\SUM_{\cF}(M) = \SUM_{\cF}(M_{|k})$.
\end{proposition}

There are three stages in which transformers produce outputs. In the first stage, the output of the decoder is linearly mapped into a desired \textbf{output dimension} $d_{\mathrm{out}}$, so every suffix vertex ends up with a vector in $\cF^{d_{\mathrm{out}}}$; we study this stage in Section \ref{section: skeleton section}. 
A final $\softmax$ can then be applied, typically with $d_\text{out}=|\Sigma\cup\{\EOS\}|$ (where $\EOS$ is a special ``End Of String''-token) to obtain a probability distribution over tokens; we examine this stage in Section \ref{section: similarity section}. 
Finally, this can be used to autoregressively generate tokens; we study this in Sections \ref{section:inference-loop} and \ref{section: cpl}.

\subsection{Logic}

The first logic we consider is \textbf{two-sorted global past temporal logic} ($\GPTLF$).\footnote{The minus in $\GPTL^-$ refers to the restriction of the modalities to the prefix and suffix respectively.} The \textbf{$\Sigma$-formulae of $\GPTLF$} are constructed according to the following grammar:
\[
  \varphi ::= \bot
  \mid p 
  \mid \neg \varphi 
  \mid (\varphi \land \varphi) 
  \mid \langle G\rangle^{\text{pre}}_{\ge k}\,\varphi 
  \mid \langle P\rangle^{\text{suf}}\,\varphi,
\]
where $p \in \Sigma$ and $k \in \mathbb{N}$.
The semantics of $\GPTLF$ is as follows.
Given a $\Sigma$-formula $\varphi$ of $\GPTLF$ and a pointed $\Sigma$-labeled two-sorted word $(\cW,i)$: 
\begin{itemize}
    \item 
    $\cW,i\models \bot$ never holds, 
    \item 
    $\cW,i\models p  \text{ iff } \lambda(i)=p$,
    \item 
    the semantics of $\neg$ and $\wedge$ are defined as usual,
    \item 
    $\cW,i\models \langle G\rangle^{\text{pre}}_{\ge k}\,\varphi \text{ iff } |\{j\in V(\cW_p)\mid \cW,j\models \varphi\}|\ge k$,~and
    \item 
    $\cW,i\models \langle P\rangle^{\text{suf}}\,\varphi \text{ iff } {\exists j \in V(\cW_s)} \text{ such that } j < i \text{ and }\ \cW,j\models \varphi$.
\end{itemize}
For a tuple $\vec\varphi = (\varphi_1, \dots, \varphi_m)$ of formulae of $\GPTLF$
and $i \in V(\cW_s)$, the \textbf{output of $\vec\varphi$ at $(\cW,i)$} is the bit string $b_1 \cdots b_m$ such that $b_\ell = 1$ iff $\cW,i \models \varphi_\ell$.

\textbf{Counting propositional logic} (CPL) 
is
the fragment of $\GPTLF$ that does not contain the past modality $\langle P \rangle ^{\text{suf}}$.
A formula of CPL is a \textbf{sentence} if all instances of proposition symbols occur under the scope of a modality $\langle G \rangle ^{\text{pre}}_{\geq k}$. We evaluate sentences of CPL over words in the natural way: 
a sentence is true iff it is true in any vertex.

\subsection{Automata}\label{subsection: automata}

We also consider a type of distributed automata operating on two-sorted words where vertices receive messages from all prefix vertices and preceding suffix vertices. Formally, a \textbf{counting past-global distributed automaton} ($\CPG$-automaton) over the vocabulary $\Sigma$ is a tuple $A = (Q, \pi, \delta, n, b)$ 
consisting of a non-empty set $Q \subseteq \{0,1\}^m$ ($m \in \Z_+$) of \textbf{states}, an \textbf{initializing function} $\pi : \Sigma \to Q$ (assigning an initial state to a vertex based on its token), a \textbf{transition function} 
$
  \delta \colon Q \times \mathcal{M}_k(Q) \times \mathcal{M}_k(Q) \to Q,
$
an \textbf{output round} $n \in \N$ and \textbf{output length} $b \in \Z_+$.

At each synchronous step, each vertex updates its state based on:
\begin{enumerate}
    \item 
    its current state,
    \item 
    the $k$-capped multiset of current states of prefix vertices, and
    \item 
    the $k$-capped multiset of current states of prior suffix vertices.
\end{enumerate}
More formally, we define the state $x_{i}^{t}$ of a vertex $i$ of a two-sorted word $\cW = (V, E, \lambda)$ after $t$ iterations of $A$ as follows. For $t = 0$, $x_i^0 := \pi(\lambda(i))$. Next, assume $x_j^t$ is defined for all $j \in V$. Let $i_0$ be the one vertex labeled $\BOS$. Then
\[
    x_i^{t+1} := \delta(x_i^t, \,\{\{\, x_j^t \mid j < i_0 \,\}\}_{|k}, \,\{\{\, x_j^t \mid i_0 \leq j < i \,\}\}_{|k}).
\]

The vertex updates stop in round $n$, i.e., after $n$ iterations of $\delta$.
The prefix of length $b$ of $x_i^n$ (with $i$ in the suffix) is called the \textbf{output of $A$ at $(\cW,i)$}. Intuitively, the truncation means that the output does not have to contain the auxiliary information used in internal computation; for example, if $b = 1$ then the output is binary: ``yes'' or ``no''.

\section{Characterizing encoder--decoder transformers via $\GPTLF$}\label{section: characterizations}

We begin our analysis by giving logical characterizations of encoder--decoder transformers in three output stages.
We start with the fundamental case where the transformers can have any output dimension but no final softmax. Building on this, we then consider the case where they do use a final softmax and output a probability distribution over tokens.
Finally, we extend even further to the case where the transformers generate tokens via autoregression.
While these results are of independent interest, they are also a necessary step toward a more sophisticated characterization of encoder--decoder transformers which we will consider in Section~\ref{section: cpl}.

\subsection{Equivalence before the final softmax}\label{section: skeleton section}

We first compare transformers, logic and automata at the level of
the pre-$\softmax$ output (the \textbf{logits}).
Let $A$ and~$B$ be classes consisting of encoder--decoder transformers, tuples of formulae of $\GPTLF$ and $\CPG$-automata, and let $a \in A$ and $b \in B$.
We say that $a$ and $b$ are \textbf{equivalent} (or \textbf{logit-equivalent}) if they produce the same outputs with the same inputs.
There are two ways to interpret the final feature vector $(\bbf_1, \dots, \bbf_{d_{\mathrm{out}}}) \in \cF^{d_{\mathrm{out}}}$ (representing the rational numbers $f_1, \dots, f_{d_{\mathrm{out}}}$) of a given vertex as a bit string. The first is simply $\bbf_1 \cdots \bbf_{d_{\mathrm{out}}}$ which we use by default. Only in this section (\ref{section: skeleton section}), when translating from logic or automata to transformers, we construct the transformer such that $f_1, \dots, f_{d_{\mathrm{out}}} \in \{0,1\}$ and instead interpret the output as $f_1 \cdots f_{d_{\mathrm{out}}}$.
We say that $A$ and $B$ have the same expressive power if for each $a \in A$ there exists an equivalent $b \in B$ and vice versa.

Translating from logic to transformers utilises a float phenomenon called 
\textbf{underflow}: any number as close or closer to zero than to any other float in the format rounds to zero.
The following proposition exemplifies this.

\begin{proposition}[\cite{ahvonen2026}]\label{proposition: underflow}
    Let $\mathcal{F}$ be a floating-point format, $f$ the smallest positive float in $\cF$, and $k$ an integer such that $\frac{k}{2}$ is accurately representable in $\cF$. Then, for all $F\in \cF$, $|F|\leq |\frac{1}{k}|$ iff $F\cdot (\frac{k}{2}\cdot f) = 0$.
\end{proposition}

Our first characterization theorem is as follows.

\skeletontheorem 
\begin{proof}
    (Sketch.) (Logic $\Rightarrow$ transformer.) The truth value of each subformula is computed in its own feature column. Connectives $\neg$ and $\land$ are local to a vertex and handled by $\MLP$s. Each
    $\langle G\rangle^{\text{pre}}_{\geq k}$ is handled by an unmasked
    self-attention or cross-attention head constructed following
    \cite{ahvonen2026}, i.e., we use Proposition~\ref{proposition: underflow} so
    that underflow occurs when $k$ or more prefix vertices satisfy
    the subformula. Each $\langle P\rangle^{\text{suf}}$ is handled by a
    masked self-attention head. The final linear map extracts the relevant truth values. Details are in Appendix~\ref{appendix: logic to transformers}.

    (Transformer $\Rightarrow$ automaton.) Each attention sub-layer and
    each $\MLP$ is simulated by a transition of the automaton. $\MLP$s are local, so
    the transition function simply acts as a look-up table. Unmasked self-attention and
    cross-attention are computed from the first multiset of the transition
    function and masked self-attention from the second. By
    Proposition~\ref{proposition: floating-point saturation}, the sums in
    an attention sub-layer depend on multiplicities only up to the
    saturation bound $k$ of $\cF$, so $k$-capped multisets suffice. Details are in
    Appendix~\ref{appendix: transformers to automata}.

    (Automaton $\Rightarrow$ logic.) For
    each state $q$ of the automaton $A$, we construct a formula $\varphi^n_q$ that holds at a
    pointed two-sorted word $(\cW,i)$ iff $A$ assigns $q$ to $i$
    in the output round $n$. The formula $\varphi^n_q$ is the disjunction of the \textbf{types} of
    all pointed words that receive $q$, where the type of $(\cW,i)$ is a
    formula describing $i$ precisely enough that every vertex satisfying it
    receives the same states as $i$ in the first $n$ rounds of every
    $\CPG$-automaton. The formula for the $m$th output bit is the
    disjunction of those $\varphi^n_q$ where the $m$th bit of $q$ is
    $1$. Details are in Appendix~\ref{appendix: automata to logic}.
\end{proof}

\subsection{Equivalence after the final softmax}\label{section: similarity section}

In the transformer architecture of Vaswani et al. \cite{vaswani2017attention}, the final linear transformation has output dimension $|\Sigma\cup\{\EOS\}|$, after which a final softmax layer is applied to generate a probability distribution over the output vocabulary at each position. We extend our characterization to this level via equivalence up to a notion of \textit{similarity}, since not all outputs of formula tuples and automata are possible outputs of the $\softmax$ function.

Fix a \textbf{similarity relation} $\sim\, \subseteq \cF^{d_{\text{out}}} \times \cF^{d_{\text{out}}}$.\footnote{Note that $\sim$ can be generalised to a relation between all bit strings of the same length as those in $\cF^{d_{\text{out}}}$ by interpreting missing strings via the corresponding normalized and subnormalized floats in $\cF$.} 
Intuitively, $\sim$ could connect feature vectors based on how close the values of the vectors are to each other, but we allow any relation $\sim$. 
Let $A$ and $B$ be classes consisting of encoder--decoder transformers with the final $\softmax$, tuples of formulae of $\GPTLF$ and $\CPG$-automata, and let $a \in A$ and $b \in B$.
We say that $a$ is \textbf{equivalent to $b$ w.r.t.\ $\sim$} if for each two-sorted word $\cW$ and each vertex $i$ in the suffix of $\cW$, the following holds: if the output of $a$ at $(\cW,i)$ is $\bv$ and there exists an output $\bu$ in the range of $b$ such that $\bv \sim \bu$, then $b$ outputs such an output $\bu$ at $(\cW,i)$. 
We say that $A$ and $B$ \textbf{have the same expressive power w.r.t.\ $\sim$} if for each $a\in A$ there exists some $b \in B$ that is equivalent to $a$ w.r.t.\ $\sim$ and vice versa.

\similaritytheorem
\begin{proof}
    We extend the translations of Theorem~\ref{theorem:main} as follows.
    From transformers to automata, the automaton first computes 
    the pre-$\softmax$ 
    output $\bb$ exactly as
    before, and then applies a fixed map on bit strings that sends $\bb$ to
    some output similar to $\softmax(\bb)$ whenever one exists.
    From automata to logic, we first construct a tuple that simulates the output of the automaton as before; these formulae can then be used to map to a similar output instead by using the connectives $\neg$ and $\land$.
    From
    logic to transformers, we compute the truth values of subformulae as before, which contains the 
    output $\bb$ of the formula
    tuple.
    We add extra decoder layers where the $\MLP$s together with the final linear map implement
    a look-up table that, whenever $\bb$ is similar to some output of
    $\softmax$, produces a pre-image of that output; the final $\softmax$
    then yields an output similar to $\bb$.
\end{proof}

\subsection{Generation-equivalence}\label{section:inference-loop}

Moving beyond the final softmax, we extend our characterization to the token generation framework by extending encoder--decoder transformers, 
tuples of $\GPTL^-$ and $\CPG$-automata 
to be autoregressive.
A \textbf{token-generating system} is a pair $(a,f)$, where
$a$ is one the above objects over $\Sigma$ and
$f$ is a \textbf{token-generating function} of the type $\{0,1\}^n\to\Sigma\cup\{\EOS\}$ where $n$ is the length of the output bit strings of $a$.\footnote{
Here the definition of $f$ allows only deterministic token generation. By changing the range of $f$ to be a power set or a set of probability distributions over $\Sigma \cup \{\EOS\}$, we could allow non-deterministic or probabilistic token generation; we leave this mostly for future work but state that Theorems~\ref{theorem: generation-equivalence} and \ref{theorem: transformer cpl generation equivalence} still hold as the proofs do not depend on the range of~$f$.
}
Let $\cW$ be a word over $\Sigma$. First we turn $\cW$ into a two-sorted word $\cW^{(1)}$ by adding a vertex labeled $\BOS$ to the end.
Tokens are generated by, at step $t > 1$, computing the output $o$ of $a$ at the last vertex of $\cW^{(t-1)}$ and extending $\cW^{(t-1)}$ by a vertex labeled $f(o)$ to obtain $\cW^{(t)}$.
The loop terminates if $f(o) = \mathrm{EOS}$. The (possibly infinite) sequence of tokens generated is the \textbf{output sequence} of $(a,f)$ on $\cW$.
Two token-generating systems are \textbf{generation-equivalent} if they produce the same output sequences with the same inputs. We say that two classes $A$ and $B$ of token-generating systems \textbf{have the same expressive power} if for each $a \in A$ there exists a generation-equivalent $b \in B$ and vice versa.

\generationtheorem 
\begin{proof}
    We extend the translations of Theorem~\ref{theorem:main} as follows. When translating between logic and automata or away from transformers, we simply use the same token-generating function $f$. From logic or automata to transformers, we construct an injection $i$ that maps the possible output strings of the token-generating system $(a,f)$ into corresponding outputs of $\softmax$ and then choose a token-generating function $f'$ such that $f(\bb) = (f' \circ i)(\bb)$ for each possible output $\bb \in \{0,1\}^n$ of $(a,f)$. We present one possible way of constructing $i$ and leave possible improvements for future work. First, we note that each $\bb$ encodes a natural number $b \leq 2^n - 1$ in binary. We add a constant number of decoder layers to the transformer constructed in the translation of Theorem~\ref{theorem:main}; together with the final linear transformation, the $\MLP$s of the added layers compute the encoding $e : \{0,1\}^n \to \cF^{|\Sigma\cup\{\EOS\}|}$ such that $e(\bb) = (b, 0, \dots, 0)$. We choose a large enough float format $\cF$ such that \textbf{(1)} each $b$ is expressible in $\cF$, and \textbf{(2)}~$\softmax$ is injective with respect to the image of $e$. The injection $i$ that we seek is $\softmax \circ e$.
\end{proof}

\section{Eliminating autoregression}\label{section: cpl}

The logical characterizations in Section~\ref{section: characterizations} make it possible to obtain logical characterizations of autoregressive transformers via logics that are not autoregressive.
We first consider the case where an autoregressive transformer is treated as a language recognizer. Building on this, we then examine the general case where we characterize the exact token sequences generated by the transformer.

\subsection{Language recognition}\label{section: characterizing accepting transformers via CPL}

We continue our analysis by considering the autoregressive encoder--decoder transformers studied in Section~\ref{section:inference-loop} in the special case of language recognition.
This is a natural framework especially in the case of encoder--decoder transformers as the input is fed into the encoder only.
An \textbf{accepting} encoder--decoder transformer is one with a distinguished acceptance token. It \textbf{accepts} an input if this token is generated at any point; otherwise it \textbf{rejects}.

As a useful tool, we define a very simple object that also recognizes languages.
For $k \in \mathbb{Z}_+$, the \textbf{$k$-capped multiset of a $\Sigma$-labeled word $\cW$}
is the multiset $M_k^\cW$ that counts the number of each token in $\cW$ up to the bound~$k$.
A \textbf{$k$-capped multiset function} is a function $f \colon \cM_k(\Sigma) \to \{0,1\}$ evaluated on $\cW$ as $f(M_k^\cW)$.

Given a language $L$ of (finite) words over $\Sigma$, we say that an accepting encoder--decoder transformer \textbf{recognizes} $L$ if it accepts exactly the words in $L$, a CPL-sentence recognizes $L$ if it is true in exactly the words in $L$, and a $k$-capped multiset function recognizes $L$ if it maps exactly the words in $L$ to $1$.

Our result is a four-way characterization of the languages recognized by transformers. It shows that accepting transformers, sentences of CPL and $k$-capped multiset functions all recognize exactly  
the \textbf{commutative star-free} languages, i.e., permutation-invariant languages definable via regular expressions with complement but without the Kleene star.

\fourwayresult 
\begin{proof}
    (Sketch.) The details of all directions are in Appendix~\ref{appendix: accepting}. 
    \textbf{(1)} From CPL to transformers, we modify the translation from the proof of Theorem \ref{theorem:main} such that the transformer outputs an accepting token if and only if the formula is true.
    \textbf{(2)} From transformers to multiset functions, we first translate the transformer into a $\CPG$-automaton via Theorem \ref{theorem: generation-equivalence} and observe that the only information gained from the input is the $k$-capped multiset of the input.
    \textbf{(3)} From multiset functions to CPL, it is easy to construct for each multiset in $\mathcal{M}_k(\Sigma)$ a corresponding sentence of CPL, and a disjunction over multisets that are mapped to $1$ gives us the desired sentence. 
    \textbf{(4)} To show that a language recognized by a multiset function is commutative and star-free, we construct a suitable regular expression that defines the language.
    \textbf{(5)} To recognize a commutative star-free language with a multiset function, we utilize Schützenberger's theorem \cite{Schutzenberger1965}.
\end{proof}

The commutative star-free languages are a strict subset of the left-deterministic polynomial languages recognized by decoder-only transformers and characterized by LTL[P] \cite{licharacterizing}. Thus, our result implies that encoder--decoder transformers are in fact less expressive than decoder-only transformers. For example, the language consisting of the words where $b$ does not precede $a$ is not permutation-invariant and thus not recognizable by any encoder--decoder transformer, but it is definable by a formula of LTL[P] and thus recognizable by a decoder-only transformer. 
This is essentially due to the fact that, in our framework, the input word is fed to the encoder rather than the decoder.

Note that Theorem~\ref{theorem: language characterization} is not dependent on the fact that our accepting transformers are autoregressive. 
We could also treat the encoder--decoder transformer as a binary classifier without autoregression, e.g.\ such that the output dimension is $1$ and the transformer accepts iff it outputs a non-negative float \cite{yang2024masked, licharacterizing}. This can be simulated by a sentence of CPL by replacing the autoregressive transformers and automata in the proof of Theorem~\ref{theorem: language characterization} with the non-autoregressive kind, and the reverse holds by Theorem~\ref{theorem:main}.

\subsection{Token generation}\label{section: characterizing generating transformers via CPL}

We now know we can express the existence of an accepting token in the generated sequence with a single sentence of CPL. This begs the question: could we express the entire output sequence in a similar way? Remarkably, the answer is yes.
Building on Theorem~\ref{theorem: generation-equivalence}, we first discover that transformers are \textbf{ultimately periodic}, meaning that for each input, either the output is finite or there exists a point at which the tokens generated begin repeating periodically. While this phenomenon is not specific to transformers---consider ultimately periodic languages and automata \cite{calbrix1993ultimately, bresolin2009theory}, and, e.g., MSO-automata \cite{Walukiewicz1996}---in our case, the periods generated are \textbf{repetition-free}, i.e., each token appears at most once per period.

\ultimatelyperiodic 
\begin{proof}
    (Sketch.)
    The proof utilizes the equivalence of transformers and $\GPTL^-$ in Theorem~\ref{theorem: generation-equivalence}. 
    We show that a formula $\varphi$ (or a sequence of formulae) of $\GPTL^-$ cannot distinguish suffix vertices that are labeled with the same token and agree on which subformulae of $\varphi$ are satisfied by prior vertices in the suffix. Thus, since the set of subformulae of $\varphi$ that are true in some prior suffix vertex is non-decreasing, eventually the next token to be generated depends only on the last generated token.
    For the details, see Appendix \ref{appendix: accepting}.
\end{proof}

Proposition~\ref{proposition: ultimately periodic} implies that the translations in Theorem~\ref{theorem: language characterization} are computable. It also means that each individual output sequence can be encoded in finite space. Thus, since Theorem~\ref{theorem: language characterization} shows that a transformer can only distinguish a finite number of inputs, we can encode all generated output sequences over all inputs in finite space.

Let $\Psi_1$ and $\Psi_2$ be matrices of CPL-sentences\footnote{The CPL-sentences could of course be replaced by $k$-capped multiset functions, accepting encoder--decoder transformers or expressions defining commutative star-free languages, and Theorem \ref{theorem: transformer cpl generation equivalence} would still hold.} with $n$ rows and let $f : \{0,1\}^n \to \Sigma \cup \{\EOS\}$. Evaluating the sentences over an input word $\cW$ yields binary matrices $\Psi_1(\cW)$ and $\Psi_2(\cW)$, and applying $f$ to their columns from left to right yields token strings $\vec\sigma_1$ and $\vec\sigma_2$.
Let $\vec\sigma_2'$ be the longest prefix of $\vec\sigma_2$ that does not contain any token twice. The output of $(\Psi_1, \Psi_2)$ with input $\cW$ is then the infinite string that starts with $\vec\sigma_1$ and is followed by infinitely many copies of $\vec\sigma_2'$, cut off to be finite after the first $\EOS$ if it contains one.

\matrixtheorem
\begin{proof}
    (Sketch.) 
    From transformers to CPL, each column of sentences reads the input $\mathcal{W}$ 
    and produces a bit string that is a pre-image of the token generated by the token-generating function of the transformer in that position.
    This is possible since the token to be generated can be inferred from the input and there is a finite number of inputs that the transformer can distinguish. From CPL to transformers, we first use decoder layers to track which column of the first matrix to simulate; once the matrix is finished, we only need to look at the current token to do the same with the second matrix indefinitely.
    For the details, see Appendix \ref{appendix: transformer cpl generation equivalence}.
\end{proof}

\section{Conclusion}\label{section: conclusion}

We have given novel logical characterizations of encoder--decoder transformers via a new temporal logic and a new type of distributed automata. The characterizations cover three different output frameworks: with and without the final softmax step of the transformer as well as autoregressive token generation.
We have then used these characterizations to obtain further results for autoregressive transformers. In terms of language recognition, we have shown that they recognize exactly the commutative star-free languages and are characterized by counting propositional logic. We have also shown that transformer outputs are ultimately periodic and each token appears at most once per period. Using this, we have characterized the exact autoregressive behaviour of the transformers via a finitary logic that is neither autoregressive nor recursive.

Finally, we stress that with minimal adjustments, our translations could cover decoder-only transformers as a special case.
With an empty prefix, the counting global modality becomes vacuous and $\GPTLF$ collapses to the logic LTL[P] by which Li \& Cotterell \cite{licharacterizing} characterize fixed-precision decoder-only transformers. We especially note that autoregressive decoder-only transformers are also ultimately periodic by a similar argument. 

\paragraph{Limitations.} Our work is purely theoretical, and no experiments were performed. 
We only considered the case without positional encodings; it would be interesting to also characterize both the sinusoidal positional encodings of \cite{vaswani2017attention} and more modern choices like rotary positional embeddings (RoPE) \cite{su2024roformer}. 
Finally, there are alternative ways to treat floats which we did not investigate, like allowing a single transformer to operate with various float formats depending on the input.

\section*{Acknowledgements}

Veeti Ahvonen was supported by the Vilho, Yrjö and Kalle Väisälä Foundation of the Finnish Academy of Science and Letters. Damian Heiman was supported by the Magnus Ehrnrooth Foundation. Antti Kuusisto, Miguel Moreno and Matias Selin were supported by the project \emph{Perspectives on computational logic}, funded by the Research Council of Finland, project number 369424.

\bibliography{literature}

@inproceedings{vaswani2017attention,
 author = {Vaswani, Ashish and Shazeer, Noam and Parmar, Niki and Uszkoreit, Jakob and Jones, Llion and Gomez, Aidan N and Kaiser, \L ukasz and Polosukhin, Illia},
 booktitle = {Advances in Neural Information Processing Systems},
 editor = {I. Guyon and U. Von Luxburg and S. Bengio and H. Wallach and R. Fergus and S. Vishwanathan and R. Garnett},
 pages = {},
 publisher = {Curran Associates, Inc.},
 title = {Attention {I}s {A}ll {Y}ou {N}eed},
 url = {https://proceedings.neurips.cc/paper_files/paper/2017/file/3f5ee243547dee91fbd053c1c4a845aa-Paper.pdf},
 volume = {30},
 year = {2017}
}

@misc{ahvonen2025expressive,
      title={Expressive {P}ower of {G}raph {T}ransformers via {L}ogic}, 
      author={Veeti Ahvonen and Maurice Funk and Damian Heiman and Antti Kuusisto and Carsten Lutz},
      year={2026},
      eprint={2508.01067},
      archivePrefix={arXiv},
      primaryClass={cs.LO},
      url={https://arxiv.org/abs/2508.01067}, 
}

@article{ahvonen2026, 
  title={{Expressive Power of Graph Transformers via Logic}}, 
  volume={40}, url={https://ojs.aaai.org/index.php/AAAI/article/view/39036}, 
  DOI={10.1609/aaai.v40i24.39036}, 
  abstractNote={Transformers are the basis of modern large language models, but relatively little is known about their precise expressive power on graphs. We study the expressive power of graph transformers (GTs) by Dwivedi and Bresson (2020) and GPS-networks by Rampásek et al. (2022), both under soft-attention and average hard-attention. Our study covers two scenarios: the theoretical setting with real numbers and the more practical case with floats. With reals, we show that in restriction to vertex properties definable in first-order logic (FO), GPS-networks have the same expressive power as graded modal logic (GML) with the global modality. With floats, GPS-networks turn out to be equally expressive as GML with the counting global modality. The latter result is absolute, not restricting to properties definable in a background logic. We also obtain similar characterizations for GTs in terms of propositional logic with the global modality (for reals) and the counting global modality (for floats).}, 
  number={24}, 
  journal={Proceedings of the AAAI Conference on Artificial Intelligence}, 
  author={Ahvonen, Veeti and Funk, Maurice and Heiman, Damian and Kuusisto, Antti and Lutz, Carsten}, 
  year={2026}, 
  month={Mar.}, 
  pages={19569-19579} 
}

@inproceedings{licharacterizing,
 author = {Li, Jiaoda and Cotterell, Ryan},
 booktitle = {Advances in Neural Information Processing Systems},
 editor = {D. Belgrave and C. Zhang and H. Lin and R. Pascanu and P. Koniusz and M. Ghassemi and N. Chen},
 pages = {159510--159554},
 publisher = {Curran Associates, Inc.},
 title = {Characterizing the {E}xpressivity of {F}ixed-Precision {T}ransformer {L}anguage {M}odels},
 url = {https://proceedings.neurips.cc/paper_files/paper/2025/file/e9e250537b0345111d50a5f8f392cffc-Paper-Conference.pdf},
 volume = {38},
 year = {2025}
}

@article{perez2021attention,
  author  = {P{\'e}rez, Jorge and Barcel{\'o}, Pablo and Marinkovic, Javier},
  title   = {Attention is {T}uring {C}omplete},
  journal = {Journal of Machine Learning Research},
  year    = {2021},
  volume  = {22},
  number  = {75},
  pages   = {1--35},
  url     = {http://jmlr.org/papers/v22/20-302.html}
}

@inproceedings{barcelo2024,
 author = {Barceló, Pablo and Kozachinskiy, Alexander and Lin, Anthony W. and Podolskii, Vladimir},
 booktitle = {International Conference on Learning Representations},
 editor = {B. Kim and Y. Yue and S. Chaudhuri and K. Fragkiadaki and M. Khan and Y. Sun},
 pages = {22077--22087},
 title = {Logical {L}anguages {A}ccepted by {T}ransformer {E}ncoders with {H}ard {A}ttention},
 url = {https://proceedings.iclr.cc/paper_files/paper/2024/file/5f0fdc1acd47431f7f3bb8ee85598cef-Paper-Conference.pdf},
 volume = {2024},
 year = {2024}
}

@inproceedings{merrill2023,
 author = {Merrill, William and Sabharwal, Ashish},
 booktitle = {Advances in Neural Information Processing Systems},
 editor = {A. Oh and T. Naumann and A. Globerson and K. Saenko and M. Hardt and S. Levine},
 pages = {52453--52463},
 publisher = {Curran Associates, Inc.},
 title = {A {L}ogic for {E}xpressing {L}og-{P}recision {T}ransformers},
 url = {https://proceedings.neurips.cc/paper_files/paper/2023/file/a48e5877c7bf86a513950ab23b360498-Paper-Conference.pdf},
 volume = {36},
 year = {2023}
}

@article{Strobl2024,
   title={What {F}ormal {L}anguages {C}an {T}ransformers {E}xpress? {A} {S}urvey},
   volume={12},
   ISSN={2307-387X},
   url={http://dx.doi.org/10.1162/tacl_a_00663},
   DOI={10.1162/tacl_a_00663},
   journal={Transactions of the Association for Computational Linguistics},
   publisher={MIT Press},
   author={Strobl, Lena and Merrill, William and Weiss, Gail and Chiang, David and Angluin, Dana},
   year={2024},
   pages={543–561} }

@inproceedings{jerad2025,
  title={Unique hard attention: A tale of two sides},
  author={Jerad, Selim and Svete, Anej and Li, Jiaoda and Cotterell, Ryan},
  booktitle={Proceedings of the 63rd Annual Meeting of the Association for Computational Linguistics (Volume 2: Short Papers)},
  pages={977--996},
  year={2025}
}

@inproceedings{yang2024masked,
 author = {Yang, Andy and Chiang, David and Angluin, Dana},
 booktitle = {Advances in Neural Information Processing Systems},
 doi = {10.52202/079017-0327},
 editor = {A. Globerson and L. Mackey and D. Belgrave and A. Fan and U. Paquet and J. Tomczak and C. Zhang},
 pages = {10202--10235},
 publisher = {Curran Associates, Inc.},
 title = {{Masked Hard-Attention Transformers Recognize Exactly the Star-Free Languages}},
 url = {https://proceedings.neurips.cc/paper_files/paper/2024/file/13d7f172259b11b230cc5da8768abc5f-Paper-Conference.pdf},
 volume = {37},
 year = {2024}
}

@inproceedings{merrill2024,
 author = {Merrill, William and Sabharwal, Ashish},
 booktitle = {International Conference on Learning Representations},
 editor = {B. Kim and Y. Yue and S. Chaudhuri and K. Fragkiadaki and M. Khan and Y. Sun},
 pages = {7690--7706},
 title = {The {E}xpressive {P}ower of {T}ransformers with {C}hain of {T}hought},
 url = {https://proceedings.iclr.cc/paper_files/paper/2024/file/1f59721c106ea80f613299039112f651-Paper-Conference.pdf},
 volume = {2024},
 year = {2024}
}

@inproceedings{ahvonen_2024,
 author = {Ahvonen, Veeti and Heiman, Damian and Kuusisto, Antti and Lutz, Carsten},
 booktitle = {Advances in Neural Information Processing Systems},
 doi = {10.52202/079017-3311},
 editor = {A. Globerson and L. Mackey and D. Belgrave and A. Fan and U. Paquet and J. Tomczak and C. Zhang},
 pages = {104205--104249},
 publisher = {Curran Associates, Inc.},
 title = {Logical {C}haracterizations of {R}ecurrent {G}raph {N}eural {N}etworks with {R}eals and {F}loats},
 url = {https://proceedings.neurips.cc/paper_files/paper/2024/file/bca7a9a0dd85e2a68420e5cae27eccfb-Paper-Conference.pdf},
 volume = {37},
 year = {2024}
}

@InProceedings{kuusisto2013,
  author =	{Kuusisto, Antti},
  title =	{{Modal Logic and Distributed Message Passing Automata}},
  booktitle =	{Computer Science Logic 2013 (CSL 2013)},
  pages =	{452--468},
  series =	{Leibniz International Proceedings in Informatics (LIPIcs)},
  ISBN =	{978-3-939897-60-6},
  ISSN =	{1868-8969},
  year =	{2013},
  volume =	{23},
  editor =	{Ronchi Della Rocca, Simona},
  publisher =	{Schloss Dagstuhl -- Leibniz-Zentrum f{\"u}r Informatik},
  address =	{Dagstuhl, Germany},
  URL =		{https://drops.dagstuhl.de/entities/document/10.4230/LIPIcs.CSL.2013.452},
  URN =		{urn:nbn:de:0030-drops-42132},
  doi =		{10.4230/LIPIcs.CSL.2013.452}
}

@misc{jiang2025softmaxtransformersturingcomplete,
      title={Softmax {T}ransformers are {T}uring-{C}omplete}, 
      author={Hongjian Jiang and Michael Hahn and Georg Zetzsche and Anthony Widjaja Lin},
      year={2025},
      eprint={2511.20038},
      archivePrefix={arXiv},
      primaryClass={cs.FL},
      url={https://arxiv.org/abs/2511.20038}, 
}

@article{su2024roformer,
  title={Ro{F}ormer: {E}nhanced {T}ransformer with {R}otary {P}osition {E}mbedding},
  author={Su, Jianlin and Ahmed, Murtadha and Lu, Yu and Pan, Shengfeng and Bo, Wen and Liu, Yunfeng},
  journal={Neurocomputing},
  volume={568},
  pages={127063},
  year={2024},
  publisher={Elsevier}
}

@inproceedings{chiang2022overcoming,
  title={Overcoming a {T}heoretical {L}imitation of {S}elf-{A}ttention},
  author={Chiang, David and Cholak, Peter},
  booktitle={Proceedings of the 60th Annual Meeting of the Association for Computational Linguistics (Volume 1: Long Papers)},
  pages={7654--7664},
  year={2022}
}

@inproceedings{li2024chain,
  title={Chain of {T}hought {E}mpowers {T}ransformers to {S}olve {I}nherently {S}erial {P}roblems},
  author={Li, Zhiyuan and Liu, Hong and Zhou, Denny and Ma, Tengyu},
  booktitle={The Twelfth International Conference on Learning Representations},
  year={2024}
}

@misc{zheng2024open,
      title={Open-{S}ora: {D}emocratizing {E}fficient {V}ideo {P}roduction for {A}ll}, 
      author={Zangwei Zheng and Xiangyu Peng and Tianji Yang and Chenhui Shen and Shenggui Li and Hongxin Liu and Yukun Zhou and Tianyi Li and Yang You},
      year={2024},
      eprint={2412.20404},
      archivePrefix={arXiv},
      primaryClass={cs.CV},
      url={https://arxiv.org/abs/2412.20404}, 
}

@inproceedings{wilkinson,
  author       = {James Hardy Wilkinson},
  title        = {Rounding {E}rrors in {A}lgebraic {P}rocesses},
  booktitle    = {Information Processing, Proceedings of the 1st International Conference
                  on Information Processing, UNESCO, Paris 15-20 June 1959},
  pages        = {44--53},
  publisher    = {{UNESCO} (Paris)},
  year         = {1959},
  bibsource    = {dblp computer science bibliography, https://dblp.org}
}

@article{floats_robertazzi,
  author       = {Thomas G. Robertazzi and
                  Stuart C. Schwartz},
  title        = {Best "{O}rdering" for {F}loating-{P}oint {A}ddition},
  journal      = {{ACM} Trans. Math. Softw.},
  volume       = {14},
  number       = {1},
  pages        = {101--110},
  year         = {1988},
  doi          = {10.1145/42288.42343},
  bibsource    = {dblp computer science bibliography, https://dblp.org}
}

@article{higham_article,
  author       = {Nicholas J. Higham},
  title        = {The {A}ccuracy of {F}loating {P}oint {S}ummation},
  journal      = {{SIAM} J. Sci. Comput.},
  volume       = {14},
  number       = {4},
  pages        = {783--799},
  year         = {1993},
  doi          = {10.1137/0914050},
  bibsource    = {dblp computer science bibliography, https://dblp.org}
}

@inproceedings{yangcounting,
title={{Counting Like Transformers: Compiling Temporal Counting Logic Into Softmax Transformers}},
author={Andy Yang and David Chiang},
booktitle={First Conference on Language Modeling},
year={2024},
url={https://openreview.net/forum?id=FmhPg4UJ9K}
}

@inproceedings{chiang2023,
author = {Chiang, David and Cholak, Peter and Pillay, Anand},
title = {Tighter {B}ounds on the {E}xpressivity of {T}ransformer {E}ncoders},
year = {2023},
publisher = {JMLR.org},
booktitle = {Proceedings of the 40th International Conference on Machine Learning},
articleno = {221},
numpages = {19},
location = {Honolulu, Hawaii, USA},
series = {ICML'23}
}

@inproceedings{Yang2025,
 author = {Yang, Andy J and Cadilhac, Micha\"{e}l and Chiang, David},
 booktitle = {Advances in Neural Information Processing Systems},
 editor = {D. Belgrave and C. Zhang and H. Lin and R. Pascanu and P. Koniusz and M. Ghassemi and N. Chen},
 pages = {37184--37225},
 publisher = {Curran Associates, Inc.},
 title = {Knee-{D}eep in {C}-{RASP}: {A} {T}ransformer {D}epth {H}ierarchy},
 url = {https://proceedings.neurips.cc/paper_files/paper/2025/file/354a6e7be4f1a35127746b8147a31bfe-Paper-Conference.pdf},
 volume = {38},
 year = {2025}
}

@article{Schutzenberger1965,
title = {{On Finite Monoids Having Only Trivial Subgroups}},
journal = {Information and Control},
volume = {8},
number = {2},
pages = {190-194},
year = {1965},
issn = {0019-9958},
doi = {https://doi.org/10.1016/S0019-9958(65)90108-7},
url = {https://www.sciencedirect.com/science/article/pii/S0019995865901087},
author = {M.P. Schützenberger}
}

@book{mcnaughton1971counter,
  title={Counter-Free Automata (MIT research monograph no. 65)},
  author={McNaughton, Robert and Papert, Seymour A},
  year={1971},
  publisher={The MIT Press}
}

@article{nllb2024scaling,
  title={Scaling neural machine translation to 200 languages},
  author={NLLB~Team,{}},
  journal={Nature},
  volume={630},
  number={8018},
  pages={841--846},
  year={2024},
  publisher={Nature Publishing Group UK London}
}

@article{sekoyan2025canary,
  title={Canary-1b-v2 \& parakeet-tdt-0.6 b-v3: Efficient and high-performance models for multilingual asr and ast},
  author={Sekoyan, Monica and Koluguri, Nithin Rao and Tadevosyan, Nune and Zelasko, Piotr and Bartley, Travis and Karpov, Nikolay and Balam, Jagadeesh and Ginsburg, Boris},
  journal={arXiv preprint arXiv:2509.14128},
  year={2025}
}

@inproceedings{yang2026probability,
  title={Probability distributions computed by autoregressive transformers},
  author={Yang, Andy and Svete, Anej and Li, Jiaoda and Lin, Anthony W and Rawski, Jonathan and Cotterell, Ryan and Chiang, David},
  booktitle={International Conference on Learning Representations},
  volume={2026},
  pages={88582--88601},
  year={2026}
}

@inproceedings{Walukiewicz1996,
  author       = {Igor Walukiewicz},
  editor       = {Claude Puech and
                  R{\"{u}}diger Reischuk},
  title        = {Monadic Second Order Logic on Tree-Like Structures},
  booktitle    = {{STACS} 96, 13th Annual Symposium on Theoretical Aspects of Computer
                  Science, Grenoble, France, February 22-24, 1996, Proceedings},
  series       = {Lecture Notes in Computer Science},
  volume       = {1046},
  pages        = {401--413},
  publisher    = {Springer},
  year         = {1996},
  doi          = {10.1007/3-540-60922-9\_33},
  bibsource    = {dblp computer science bibliography, https://dblp.org}
}

@article{bresolin2009theory,
  title={A theory of ultimately periodic languages and automata with an application to time granularity},
  author={Bresolin, Davide and Montanari, Angelo and Puppis, Gabriele},
  journal={Acta Informatica},
  volume={46},
  number={5},
  pages={331--360},
  year={2009},
  publisher={Springer}
}

@inproceedings{calbrix1993ultimately,
  title={Ultimately periodic words of rational $\omega$-languages},
  author={Calbrix, Hugues and Nivat, Maurice and Podelski, Andreas},
  booktitle={International Conference on Mathematical Foundations of Programming Semantics},
  pages={554--566},
  year={1993},
  organization={Springer}
}

\cleardoublepage

\appendix

\section{Additional related work}\label{appendix: related work}

In this section, we cover related work pertaining to upper and lower bounds on the expressivity of transformers as well as chain-of-thought models.

Regarding lower and upper bounds for the expressive power of transformers, Chiang et al. \cite{chiang2023} established $\mathsf{FOC[+;MOD]}$ as an upper bound for fixed-precision encoder-only transformers and as a lower bound for the case with reals. Merrill \& Sabharwal \cite{merrill2023} showed that first-order logic with majority quantifiers is an upper bound for log-precision transformers. \cite{yangcounting} established Minimal Tense Logic with counting terms as a lower bound for masked soft-attention transformers with unbounded input size. Yang \& Chiang \cite{barcelo2024} gave first-order logic with unary numerical predicates as a lower bound for unique hard-attention transformers with reals as well as linear temporal logic with unary numerical predicates and counting formulae as a lower bound for average hard-attention.

Another direction has been to investigate the effect of \textit{chain-of-thought} (CoT) on decoder-only models, i.e., allowing the model to generate ``scratchpad'' tokens that are discarded before the final answer.
Merrill \& Sabharwal \cite{merrill2024} began in the log-precision setting, showing that CoT increases expressive power depending on the number of intermediate steps, a polynomial number capturing exactly the complexity class $\mathsf{P}$. Li et al. \cite{li2024chain} showed that with finite-precision and a constant amount of decoder layers, CoT allows constant-depth transformers to solve inherently serial problems that would otherwise require proportionally many layers. Jiang et al. \cite{jiang2025softmaxtransformersturingcomplete} showed that length-generalizable soft-attention transformers with CoT are Turing-complete, where length-generalizability means that an idealized learning procedure will converge to the language expressible with the transformer model when given all words of some length in the language.

\section{Detailed definitions for encoder--decoder transformers}\label{appendix: encoderdecoder}

In this section, we give a more detailed definition of encoder--decoder transformers.

First, some notations. For a matrix $X$, we let $X_{i,*}$ denote the $i$th row, $X_{*,j}$ the $j$th column, and $X_{i,j}$ the element in the $i$th row and $j$th column of $X$.
For a set $S$, let $S^*$ denote the set of finite sequences over $S$, let $S^{* \times *}$ denote the set of matrices over $S$ and let $S^{m \times *}$ and $S^{* \times n}$ denote the restriction of $S^{* \times *}$ to matrices with $m$ rows and $n$ columns, respectively. 
For matrices $X^{(1)}, \dots, X^{(h)} \in S^{m \times n}$, we let $\concat(X^{(1)}, \dots, X^{(h)})$ denote the matrix $Y$ where $Y_{i,(\ell-1)n + j} = X^{(\ell)}_{i,j}$.

We define an encoder--decoder transformer over a floating-point format $\mathcal{F}$ as follows.

\paragraph{Embedding.} 
First, an \textbf{embedding function} of dimension $d$ is a function $\mathrm{em}:\Sigma\to \cF^d$ that maps each token to a vector in the embedding space. Applying it vertex-wise to the prefix $\mathcal{W}_p$ and suffix $\mathcal{W}_s$ thus produces matrices in $\cF^{|\mathcal{W}_p|\times d}$ and $\cF^{|\mathcal{W}_s|\times d}$, where the $i$th row of the former is $\mathrm{em}(\lambda(i))$ and the $j$th row of the latter is $\mathrm{em}(\lambda(|\cW_p| + j))$.

\paragraph{Masking.}

The \textbf{strict causal mask} maps each square matrix $S\in\cF^{*\times *}$ to the matrix $\mathrm{mask}(S)$ of the same dimensions with the upper-right-triangular entries set to $-\infty$. In other words, $\mathrm{mask}(S)_{ij} := -\infty$ if $j\geq i$ and
$\mathrm{mask}(S)_{ij} := S_{ij}$ otherwise. For non-strict masking and the shifting operation of Vaswani et al.  \cite{vaswani2017attention}, see Appendix \ref{appendix:masking}.

\paragraph{Softmax.}

The softmax function intuitively turns each vector into a probability distribution. More formally, for each vector $\bx = (x_1, \dots, x_{n}) \in \cF^{*}$, $\softmax(\bx) \in \cF^*$ is the vector of the same length as $\bx$ such that
$
    \softmax(\bx)_i := \frac{e^{x_i - \max(\bx)}}{\sum_{j = 1}^n e^{x_j - \max(\bx)}}
$
where the bias 
$\max(\bx)$
is subtracted for numerical stability.

\paragraph{Attention.}
An \textbf{attention head} with dimensions $(d, d_k, d_v)$ consists of
parameters $W^Q,W^K\in\cF^{d\times d_k}$ and $W^V\in\cF^{d\times d_v}$, and maps a pair of matrices
$X,Y\in\cF^{*\times d}$ to a matrix with as many rows as $Y$ defined as
\[
    H_{\mathrm{mask}}(X,Y) := \softmax\Big(\mathrm{mask}\Big(
    \frac{(YW^Q)(XW^K)^{\top}}{\sqrt{d_k}}\Big)\Big)(XW^V)\in\cF^{*\times d_v},
\]
where $\softmax$ is applied row-wise.
Intuitively, $X$ and $Y$ are preceding outputs obtained from the encoder and decoder respectively.
We consider three special cases of attention heads: the \textbf{masked self-attention head} $H_{\mathrm{mask}}(Y):=H_{\mathrm{mask}}(Y,Y)$\footnote{After masking and before $\softmax$, the top row is $(-\infty, \dots, -\infty)$. If $\softmax$ is computed normally, this results in a row of NaNs. To avoid this, we adopt the convention that $\softmax(-\infty, \dots, -\infty) := (0, \dots, 0)$. This is a natural choice 
that 
simply stops the $\BOS$-token from making predictions about future tokens.}, 
and two cases where masking is omitted:
the \textbf{cross-attention head} $H(X,Y)$ and the \textbf{unmasked self-attention head} $H(X):=H(X,X)$.

A \textbf{multi-head attention layer of width $h$} combines heads $H^{(1)},\dots,H^{(h)}$ of one of the above kinds with dimensions $(d,d_k,d_v)$ with a parameter matrix $W^O\in\cF^{hd_v\times d}$. For unmasked self-attention: 
\[
    \MHSA(X) := \concat(H^{(1)}(X),\dots,H^{(h)}(X))\,W^O\in \cF^{*\times d},
\]
and analogously $\MHSA(Y)_{\mathrm{mask}}$ for masked self-attention and $\MHCA(X,Y)$ for cross-attention; the output has $d$ columns and as many rows as $X$ (resp.\ $Y$). Multiple heads do not increase expressivity (see Appendix \ref{appendix:multi-head-cross-attention}), but makes some proofs easier.

\paragraph{MLPs.}

An \textbf{MLP} (multi-layer perceptron) with dimensions $(d, d_{ff})$ is a function $M : \cF^{1 \times d} \to \cF^{1 \times d}$ defined w.r.t. matrices $W_1 \in \cF^{d \times d_{ff}}$ and $W_2 \in \cF^{d_{ff} \times d}$ and bias vectors $b_1 \in \cF^{1 \times d_{ff}}$ and $b_2 \in \cF^{1 \times d}$. For each (horizontal) vector $\bx \in \cF^{1 \times d}$, we define
\[
    M(\bx) := \ReLU(\bx W_1 + b_1) W_2 + b_2,
\]
where $\ReLU(x) = \max\{0,x\}$ is applied element-wise.
For a matrix $X \in \cF^{* \times d}$, $M(X)$ is obtained by applying $M$ separately to each row of $X$.

\paragraph{Layers.} An \textbf{encoder layer} $\EL$ and a
\textbf{decoder layer} $\DL$ with dimensions $(d,d_k,d_v,d_{ff})$ combine attention layers and MLPs as follows 
where the inputs and outputs are in $\cF^{* \times d}$, $\EL(X)$ has as many rows as $X$ and $\DL(X,Y)$ has as many rows as $Y$:
\begin{align*}
    \EL(X) &:= M(X')+X' &&\text{where } X' := \MHSA(X)+X,\\
    \DL(X,Y) &:= M(Y'')+Y'' &&\text{where } Y'' := \MHCA(X,Y')+Y'
    \text{ and } Y' := \MHSA_{\mathrm{mask}}(Y)+Y,
\end{align*}
each sub-layer having its own parameters. The sums are called \textbf{residual connections}. Layer normalization, commonly applied before or after the
residual connections, does not affect expressivity and is omitted (see Appendix~\ref{appendix:layer-normalization}). An \textbf{encoder $\mathsf{E}$} and \textbf{decoder $\mathsf{D}$ of length $L$} are compositions of $L$ such layers:
\begin{align*}
    &\mathsf{E}(X):=(\EL_L\circ\dots\circ \EL_1)(X),\\
    &\mathsf{D}(X, Y):=
    \DL_L(X,\DL_{L-1}(X, \cdots \DL_2(X, \DL_1(X,Y))\cdots)).
\end{align*}

\paragraph{Linear transformation.}
Given the final decoder output $Y \in \cF^{|\mathcal{W}_s| \times d}$ and a fixed output dimension $d_\text{out}$, a final 
\textbf{linear transformation} with dimensions $(d_\text{},d_\text{out})$ with parameter $W^{\mathrm{out}} \in \cF^{d \times d_\text{out}}$ maps the output to
\[
    \text{Out}(Y):=Y\cdot  W^{\mathrm{out}}\in \cF^{|\mathcal{W}_s|\times d_\text{out}}.
\]

\paragraph{Encoder--decoder transformer. } 

An \textbf{encoder--decoder
transformer} $T$ with the dimensions $(d,d_k,d_v,d_{ff},d_{\mathrm{out}})$ consists of an
embedding function $\mathrm{em}$ of dimension $d$, an encoder $\mathsf{E}$
and decoder $\mathsf{D}$ 
of dimensions $(d,d_k,d_v,d_{ff})$, and a linear transformation $\text{Out}$ with dimensions $(d,d_\text{out})$. $T$ maps each $\Sigma$-labeled two-sorted word $\mathcal{W}$ to
\[
    \text{T}(\mathcal{W}):=\text{Out}(\mathsf{D}(\mathsf{E}(\mathrm{em}(\mathcal{W}_p)),\mathrm{em}(\mathcal{W}_s)))\in \cF^{|\cW_s| \times d_\text{out}}.
\]

At each vertex $v$ in the suffix $\cW_s$, an encoder--decoder transformer $T$ outputs a floating-point vector $\bv = (f_1, \dots, f_{d_\text{out}})~\in~\cF^{d_{\text{out}}}$. 
This can be interpreted as a bit string in two natural ways. 
\begin{itemize}
    \item \textbf{Bit-wise interpretation:}
    We can always interpret $\bv$ as the concatenated bit string $f_1 \cdots f_{d_{\text{out}}}$.
    \item \textbf{Feature-wise interpretation:} If $f_i$ represents either $0$ or $1$ for each $i \in [d_{\text{out}}]$, then we can also interpret $\bv$ as the bit string $b_1 \cdots b_{d_{\text{out}}}$ where $b_i = 1$ if and only if $f_i$ represents $1$.
\end{itemize}

Let $a$ be either a tuple of $\GPTL^-$-formulae or a $\CPG$-automaton, and let $T$ be an encoder--decoder transformer. 
\begin{itemize}
    \item We say that $T$ is \textbf{feature-wise equivalent} to $a$ if for each pointed two-sorted graph $(\cW,i)$ the output of $T$ at $(\cW,i)$  interpreted feature-wise is the same as the output of $a$ at $(\cW,i)$.
    \item We say that $a$ is \textbf{bit-wise equivalent} to $T$ if for each pointed two-sorted graph $(\cW,i)$ the output of $a$ at $(\cW,i)$ is the same as the output of $T$ at $(\cW,i)$ interpreted bit-wise.
\end{itemize}

\section{Proofs for Theorem \ref{theorem:main}}

In this section, we prove Theorem~\ref{theorem:main}. We split the proof into three parts: in each part, we prove one of the following three statements.

\begin{restatable}[Logic $\Rightarrow$ transformers]{theorem}{logictotransformers}\label{theorem: logic to transformers}
For each tuple of formulae of $\GPTLF$, we can construct a feature-wise equivalent encoder--decoder transformer without the final $\softmax$.
\end{restatable}

\begin{restatable}[Transformers $\Rightarrow$ automata]{theorem}{transformerstoautomata}\label{theorem: transformers to automata}
For each floating-point encoder--decoder transformer without the final $\softmax$, we can construct a bit-wise equivalent $\CPG$-automaton.
\end{restatable}

\begin{restatable}[Automata $\Rightarrow$ logic]{theorem}{automatatologic}\label{theorem: automata to logic}
For each $\CPG$-automaton, we can construct an equivalent tuple of formulae of $\GPTLF$.
\end{restatable}

\subsection{Proof of Theorem~\ref{theorem: logic to transformers}}\label{appendix: logic to transformers}

In this section, we prove Theorem~\ref{theorem: logic to transformers}.

\logictotransformers*
\begin{proof}
    We intuitively use the $\MLP$s to simulate Boolean operators, unmasked self-attention layers and/or cross-attention layers to simulate modalities $\langle G \rangle_{\geq k}^{\text{pre}}$ and masked self-attention layers to simulate modalities $\langle P \rangle^{\text{suf}}$.

    Let $\vec\varphi = (\varphi_1, \dots, \varphi_m)$ be a tuple of formulae of $\GPTLF$. Let $\psi_1, \dots, \psi_d$ enumerate the subformulae of $\varphi_1, \dots, \varphi_m$ including all tokens and $\top$ and $\bot$ such that $\psi_i = \varphi_i$ for each $i \in [m]$. We basically use $d$ as the internal dimension of our transformer, but there is no harm in padding the internal dimension up to, say, an integer multiple of $d$ to avoid problems resulting from residual connections. This can be used to ensure that the inputs and outputs of sub-layers do not interfere with each other by shifting the output to a different section of the feature dimension while unnecessary input values can be negated by the $\MLP$s; the details of shifting are already covered in \cite{ahvonen2025expressive}. For the embedding function, we map each token $p_i$ to the feature vector in $\cF^d$ such that the feature in position $j$ is $1$ if $\psi_j \in \{p_i, \top\}$ and $0$ otherwise. We then move to compute the truth values of subformulae other than tokens, $\top$ and $\bot$, one-by-one. The method of computing the truth value will depend on the connective used to construct the subformula in question.

    \paragraph{Simulating Boolean operators:} 
    For Boolean operators $\neg$ and $\land$, we use the $\MLP$s in decoder layers. We set all weights in all attention sub-layers to zero; because of the residual connection, this means that these attention sub-layers only achieve an identity transformation. For the $\MLP$s, we recall from Lemma~38 of Ahvonen et al. \cite{ahvonen2025expressive} that a $\ReLU$-activated $\MLP$ with arbitrarily many layers can simulate Boolean operators in the sense that columns of the feature matrix correspond to subformulae and the truth values of more complex subformulae are computed into their own columns based on the columns corresponding to their immediate subformulae.
    Thus, a stack of encoder or decoder layers with two-layer $\MLP$s can do the same.

    \paragraph{Simulating $\langle P \rangle^{\text{suf}}$:} 
    We use masked self-attention sub-layers to simulate $\langle P \rangle^{\text{suf}}$ as follows. Assume that column $i$ encodes the truth values of $\varphi$; we want to calculate the truth value of $\langle P \rangle^{\text{suf}} \varphi$. Let $W^Q$ and $W^K$ be $(d \times 1)$-matrices where all elements are $0$s. Now $\frac{(X W^Q) (X W^K)^{\T}}{\sqrt{d_k}}$ is the $n \times n$ matrix where every element is $0$.  

    After masking, every element on and to the right of the diagonal becomes $-\infty$. 
    Biasing does not affect any of the rows.
    After $\softmax$, all values to the left of the diagonal are non-zero while all the values to the right of the diagonal are $0$. 
    Once we define $W^V$ as the $(d \times 1)$-matrix where the $i$th row is $1$ and others are $0$s, the attention head outputs an $(n \times 1)$-matrix where each row is either some non-zero number if $\langle P \rangle^{\text{suf}}\varphi$ is true in the corresponding vertex and $0$ if $\langle P \rangle^{\text{suf}}\varphi$ is false instead. We can then use the MLPs in the decoder layers to normalize all non-zero values to $1$ while keeping all $0$s as $0$s; see Lemma 40 in \cite{ahvonen2025expressive}.

    \paragraph{Simulating $\langle G \rangle_{\geq k}^{\text{pre}}$:} 
    It was already shown in Theorem 47 in \cite{ahvonen2025expressive} how a stack of encoder layers can simulate $\langle G \rangle_{\geq k}^{\text{pre}}$. However, we also need to be able to simulate them with a stack of decoder layers since modalities $\langle G \rangle_{\geq k}^{\text{pre}}$ can be stacked on top of $\langle P \rangle^{\text{suf}}$; since we used decoder layers to simulate $\langle P \rangle^{\text{suf}}$, this means that the decoder layers must be able to handle both $\langle G \rangle_{\geq k}^{\text{pre}}$ and $\langle P \rangle^{\text{suf}}$. 

    Assume that $X$ and $Y$ are the Boolean input matrices to the cross-attention sub-layer where the $i$th column encodes the truth value of $\varphi$ and the $j$th column encodes $\top$ (i.e., it is a column of $1$s). The matrix $W^Q$ has a single column where the $j$th row is $F$ and others are zero, where $F$ is the greatest non-infinite float in the floating-point format. 
    Thus, $YW^Q$ is a single-column matrix of $F$s.
    The matrix~$W^K$ has likewise a single column where the $i$th row is $1$ and others are $0$. 
    Hence, the matrix $XW_K$ has a single column which encodes the truth values of $\varphi$.
    Now $(YW^Q)(XW^K)^\T$ is a matrix where each row is the $i$th column of $X$ multiplied by $F$. This is otherwise the same outcome as in the construction from \cite{ahvonen2025expressive, ahvonen2026}, except there are no rows of zeros unless $\varphi$ is true in zero vertices. This is strictly beneficial as it only eliminates the possibility of false positives. The rest of the construction follows the same scheme as \cite{ahvonen2025expressive, ahvonen2026}; the idea is to construct $W_V$ in a simple way that triggers underflow (see Proposition~\ref{proposition: underflow}) exactly when the number of vertices satisfying $\varphi$ is at least $k$. This only works for half of the possible values of~$k$ but, as covered in \cite{ahvonen2025expressive}, the remaining cases can be covered with additional arithmetical analysis and the $\MLP$s can be used to normalize all values back into $0$s and $1$s correctly.

    The final linear transformation of the transformer simply projects the first $m$ elements of the feature vectors, giving us the simulated truth values of $\varphi_1, \dots, \varphi_m$.
\end{proof}

\subsection{Proof of Theorem~\ref{theorem: transformers to automata}}\label{appendix: transformers to automata}

In this section, we prove Theorem~\ref{theorem: transformers to automata}. 

\transformerstoautomata*
\begin{proof}
    We construct a $\CPG$-automaton that simulates the whole transformer such that we only need to perform a single state transition per attention sub-layer and $\MLP$. All the internal arithmetic is essentially encoded into the transition function.

    Let $T$ be an encoder--decoder transformer of dimension $d$ over the format $\cF(p,q)$ with $N_1$ encoder layers and $N_2$ decoder layers. We construct a bit-wise equivalent $\CPG$-automaton $(Q, \pi, \delta, 2N_1 + 3N_2 + 1, d)$ where $Q = \{0,1\}^{d(p+q+1)+m}$ where the first $d(p+q+1)$ bits are used to simulate $T$ while the last $m$ bits (with a suitably large choice of $m$) encode which sub-layer of which encoder or decoder layer we are simulating: $0\cdots 0$ for the self-attention sub-layer of the first encoder layer, $0 \cdots 01$ for the $\MLP$ of the first encoder layer, $0 \cdots 010$ for the next attention sub-layer, etc.

    Let $\mathrm{em} : \Sigma \to \cF^{d}$ be the embedding function of the transformer $T$. This can be identified with a function $f : \Sigma \to \{0,1\}^{d(p+q+1)}$ where the first $p+q+1$ bits encode the first float in the output of $\mathrm{em}$, the next $p+q+1$ bits encode the second one, and so on. We define that $\pi(p_i)$ is $\mathrm{em}(p_i)$ concatenated with the string $0 \cdots 0$, indicating that the calculation begins with an unmasked self-attention head.

    Finally, we define the transition function; we do this based on which sub-layer we are simulating. For the remainder of the proof, we will omit the last $m$ bits of the states of the automaton, since the last $m$ bits simply tick upward in a lexicographic order. We will further simplify notation by writing $q_{\bv}$ to denote the state corresponding to the feature vector $\bv \in \cF^d$. We also simplify the setting with one more core observation; each vertex can tell whether it is in the encoder-part or decoder-part of the word based on the size of the second multiset in the transition function; thus, it is easy to define that during encoder layers, decoder vertices do not update, and during decoder layers, encoder vertices do not update.

    \paragraph{Simulating unmasked self-attention:}

    We analyse the arithmetic operations in the attention sub-layer and analyse the dependencies to show that a single transition suffices to simulate the whole layer.
    
    Assume for each vertex $i$ and each feature vector $\bv$ that $i$ is in state $q_{\bv}$ if and only if $i$ has the feature vector $\bv$ before applying the self-attention module. We focus on the simpler case where there is only one attention head (i.e., $h = 1$) and each feature vector is a single floating-point number; the technique we use generalises easily to the case with more attention heads and longer feature vectors. 
    Letting $\bv_i$ denote the feature vector of vertex $i$ (in this case a single float) before multiplication with $W^Q$ and $\bv_i^Q$ after, and letting $Q_j$ denote the $j$th column of $W^Q$ (in this case a single float), we have 
    \[
        \bv_i^Q = (\bv_i \cdot Q_1, \dots, \bv_i \cdot Q_{d_k}).
    \]
    The vectors $\bv_i^K$ and $\bv_i^V$ are calculated analogously, and each of them only depends on $\bv_i$ (and not on any $\bv_j$ where $j \neq i$).

    Next, we compute $F_{ij}^{QK}$, which will be the float in row $i$ and column $j$ of $(XW^Q)(XW^K)^{\T}$. We get 
    \[
        F_{ij}^{QK} = \sum_{k = 1}^{d_k} \bv_{i,k}^Q \bv_{j,k}^K.
    \]
    Clearly, each $F_{ij}^{QK}$ only depends on $\bv_{i}^Q$ and $\bv_{j}^K$.

    Next, we compute $F_{ij}^{\sqrt{d_k}}$, which will be the float in row $i$ and column $j$ of $\frac{(XW^Q)(XW^K)^{\T}}{\sqrt{d_k}}$:
    \[
        F_{ij}^{\sqrt{d_k}} = \frac{F_{ij}^{QK}}{\sqrt{d_k}}.
    \]

    Next, we compute $F_{ij}^{\softmax}$, which will be the float in row $i$ and column $j$ of $\softmax\left(\frac{(XW^Q)(XW^K)^{\T}}{\sqrt{d_k}}\right)$. Let $i_{\max} := \max\{\, F_{ij}^{\sqrt{d_k}} \mid 1 \leq j \leq n \,\}$. We get
    \[
        F_{ij}^{\softmax} = \frac{e^{F_{ij}^{\sqrt{d_k}} - i_{\max}}}{\sum_{k = 1}^{n} e^{F_{ik}^{\sqrt{d_k}} - i_{\max}}}.
    \]
    Due to Proposition~\ref{proposition: floating-point saturation}, $F_{i}^{\softmax}$ depends only on $F_{ik}^{\sqrt{d_k}}$ for a bounded number of $k$-values.

    Next, we compute $F_{ij}$, which will be the float in row $i$ and column $j$ of the output matrix of the attention head:
    \[
        F_{ij} = \sum_{k = 1}^{n} F_{ik}^{\softmax} \bv_{k,j}^V.
    \]
    Again, due to Proposition~\ref{proposition: floating-point saturation}, each $F_{ij}$ depends on products $F_{ik}^{\softmax} \bv_{k,i}^V$ for only a bounded number of $k$-values.

    Finally, we consider the final output matrix $O$ of the attention module. We have 
    \[
        O_{ij} = \sum_{k = 1}^{hd_v} F_{ik} W^O_{kj}.
    \]
    Due to the number of attention heads $h$ being fixed, $O_{ij}$ only depends on $F_{ik}$ for a bounded number of $k$-values.

    To show that the automaton $A$ can compute $O_i$ for each $i$ in a single step, we only need to show that $O_{i}$ depends on $\bv_j$ for only a bounded number of $j$-values. Clearly $O_{i}$ only depends on $F_{i}$. Since $F_{ij}$ depends on products $F_{ik}^{\softmax} \bv_{k,i}^V$ for only a bounded number of $k$-values, so do $F_i$ and thus $O_i$ because the number of columns $j$ is fixed. 

    For each float, there are only a finite number of pairs of floats that, when multiplied, will result in the float in question; thus, $O_i$ depends on $F_{ik}^{\softmax}$ and $\bv_{k,i}^V$ for only a bounded number of $k$-values.
    
    Because $\bv_{k,i}^V$ only depends on $\bv_k$, we only need to further consider $F_{ik}^{\softmax}$. Because $F_{i}^{\softmax}$ depends only on $F_{ik}^{\sqrt{d_k}}$ for a bounded number of $k$-values (as stated earlier), the same is true for $O_{i}$. Because $F_{ik}^{\sqrt{d_k}}$ only depends on $F_{ik}^{QK}$, $O_i$ only depends on $F_{ik}^{QK}$ for a bounded number of $k$-values. Finally, since $F_{ik}^{QK}$ only depends on $\bv_i^Q$ and $\bv_k^K$ which only depend on $\bv_i$ and $\bv_k$ respectively, $O_i$ only depends on $\bv_k$ for a bounded number of $k$-values. Thus, the automaton $A$ can calculate the transition for each vertex in a single step by only looking at the multiset of states $q_{\bv}$ of vertices in the word.
    Moreover, since the residual connection does not add any dependencies, we can skip past $O_i$ and calculate the output of the residual connection in the single transition step.

    \paragraph{Simulating $\MLP$s:}
    
    Because $\MLP$s perform only local transformations, we construct the transition function to act as a look-up table for the outputs of the $\MLP$. 
    
    Since the weight matrices and biases are constant, for each $\bv$ there exists a unique feature vector $\bv'$ that is the output of the $\MLP$ when given the input $\bv$. Thus, we define $\delta(q_{\bv}, M, M') = q_{\bv'}$.

    \paragraph{Simulating masked self-attention sub-layers:}

    We modify the case for unmasked self-attention to show that in the case of masked self-attention, the transitions depend only on the feature vectors of a bounded number of preceding vertices in the suffix of the word (and not at all on vertices that follow, or on vertices in the prefix). 
    
    First, we replace each occurrence of $X$ with $Y$, as masked self-attention occurs in decoder layers. 
    The core difference occurs in the step where we compute $F_{ij}^{\softmax}$; in the particular case of masked self-attention, we have $i_{\max} := \max\{\, F_{ij}^{\sqrt{d_k}} \mid 1 \leq j \leq i \,\}$ and the denominator only sums up to the bound $j$ rather than $n$, i.e., we have
    \[
        F_{ij}^{\softmax} = \frac{e^{F_{ij}^{\sqrt{d_k}} - i_{\max}}}{\sum_{k = 1}^{j} e^{F_{ik}^{\sqrt{d_k}} - i_{\max}}}
    \]
    for all $j \leq i$ and $F_{ij}^{\softmax} = 0$ for all $j > i$.

    The analysis on dependencies is the same as before with the following additional observations. Because $F_{ij}^{\softmax} = 0$ for all $j > i$, $O_{i}$ depends on products $F_{ik}^{\softmax} \bv_{k,i}^V$ for only a bounded number of values $k \leq i$, and it thus only depends on $\bv_{k,i}^V$ for a bounded number of values $k \leq i$. Likewise, because $F_{ij}^{\softmax}$ only depends on $F_{ik}^{\sqrt{d_k}}$ for a bounded number of values $k \leq i$, so does $O_i$, and this upper limit $i$ cascades in the analysis until we reach $\bv_k^K$. Thus, $O_i$ only depends on $\bv_k'$ for a bounded number of values $k \leq i$ and not at all for values $k > i$, and so we can read the transition from the second multiset in the same way as we did with the first multiset in the prior case. 

    \paragraph{Simulating cross-attention sub-layers:}

    We modify the case for unmasked self-attention to show that in the case of cross-attention, the transitions depend only on the feature vectors of a bounded number of vertices in the encoder part of the word (and not at all on other vertices in the decoder part). Naturally, we replace $X$ with $Y$ in the product $XW^Q_i$. The dependence of $O_{i}$ on $\bv^Q_k$ is only for the value $k = i$, and the transition function can read this from the first input. The rest of the analysis remains unchanged.

    \paragraph{Simulating the final linear transformation:}

    Similar to $\MLP$s, the final linear transformation can be calculated in a single step just by looking at the previous state of a vertex. This concludes the construction.
\end{proof}

\subsection{Proof of Theorem~\ref{theorem: automata to logic}}\label{appendix: automata to logic}

In this section, we prove Theorem~\ref{theorem: automata to logic}.

\automatatologic*

We start by defining types. Intuitively, these are formulae of $\GPTLF$ that specify a pointed word uniquely up to some modal depth and width. These are analogous to graded types in \cite{ahvonen_2024} and types in \cite{kuusisto2013}.

Before defining types, we need some auxiliary definitions. 
We define the abbreviated formula $\langle P \rangle_{\geq k}^{\text{suf}} \varphi$ recursively as follows. For $k = 1$, we define $\langle P \rangle_{\geq 1}^{\text{suf}} \varphi := \langle P \rangle^{\text{suf}} \varphi$. Assume we have defined $\langle P \rangle_{\geq k}^{\text{suf}} \varphi$. We define
\[
    \langle P \rangle_{\geq k+1}^{\text{suf}} \varphi := \langle P \rangle^{\text{suf}} (\varphi \land \langle P \rangle_{\geq k}^{\text{suf}} \varphi).
\]

Let $\GPTLF_{\geq}$ be the logic obtained from $\GPTLF$ by replacing $\langle P \rangle^{\text{suf}}$ in the grammar with $\langle P \rangle^{\text{suf}}_{\geq k}$. More formally, a $\Sigma$-formula $\varphi$ of $\GPTLF_{\geq}$ is constructed according to the following grammar:
\[
  \varphi ::= \bot
  \mid p 
  \mid \neg \varphi 
  \mid (\varphi \land \varphi) 
  \mid \langle G\rangle^{\text{pre}}_{\ge k}\,\varphi 
  \mid \langle P\rangle^{\text{suf}}_{\geq k}\,\varphi,
\]
where $p \in \Sigma$ and $k \in \N$.

Since the modality $\langle P\rangle^{\text{suf}}_{\geq k}$ can be defined as an abbreviation in $\GPTLF$ and since $\langle P\rangle^{\text{suf}}_{\geq 1}$ has the same semantics as $\langle P\rangle^{\text{suf}}$, the two logics $\GPTLF$ and $\GPTLF_{\geq}$ have the same expressive power: each formula of $\GPTLF$ can be translated into a formula of $\GPTLF_{\geq}$ and vice versa.

The \textbf{modal depth} $\md(\varphi)$ and \textbf{width} $\width(\varphi)$ of a formula $\varphi$ of $\GPTLF_{\geq}$ are defined recursively as follows:
\[
\begin{array}{rcl@{\qquad}rcl}
\md(\bot) &=& 0 & \width(\bot) &=& 0 \\
\md(p) &=& 0 & \width(p) &=& 0 \\
\md(\neg\psi) &=& \md(\psi) & \width(\neg\psi) &=& \width(\psi) \\
\md(\psi \land \theta) &=& \max\{\md(\psi),\md(\theta)\} & \width(\psi \land \theta) &=& \max\{\width(\psi),\width(\theta)\} \\
\md(\langle G\rangle^{\text{pre}}_{\geq k}\psi) &=& \md(\psi)+1 & \width(\langle G\rangle^{\text{pre}}_{\geq k}\psi) &=& \max\{\width(\psi),k\} \\
\md(\langle P\rangle^{\text{suf}}_{\geq k}\psi) &=& \md(\psi)+1 & \width(\langle P\rangle^{\text{suf}}_{\geq k}\psi) &=& \max\{\width(\psi),k\}
\end{array}
\]

Finally, we define some auxiliary formulae. Let $\varphi$ be a formula of $\GPTLF$ or $\GPTLF_{\geq}$. We define
\[
    \langle G \rangle_{=k}^{\text{pre}} \varphi := \langle G \rangle_{\geq k}^{\text{pre}} \varphi \land \neg \langle G \rangle_{\geq k+1}^{\text{pre}} \varphi.
\]
We also define $\langle P \rangle_{= k}^{\text{suf}} \varphi := \langle P \rangle_{\geq k}^{\text{suf}} \varphi \land \neg \langle P \rangle_{\geq k+1}^{\text{suf}} \varphi$.

\begin{definition}\label{type_definition}
    The \textbf{$\GPTLF_{\geq}$-type of depth $d \in \N$ and width $k \in \N$} of a pointed word $(\cW,i)$ where $\cW = (V, E, \lambda)$, denoted by $\tau^{(\cW,i)}_{k,d}$, is defined recursively as follows. For $d = 0$, we define 
    \[
        \tau^{(\cW,i)}_{k,0} := \bigwedge_{\lambda(i) = p} p \land \bigwedge_{\lambda(i) \neq p} \neg p.
    \]
    Assume we have defined $\tau^{(\cW,i)}_{k,d}$ for each pointed word $(\cW,i)$ and let $T_{k,d}$ denote the set of all such types. We define $\tau^{(\cW,i)}_{k,d+1}$ intuitively as follows
    \[
        \begin{aligned}
            \tau^{(\cW,i)}_{k,d+1} := \tau^{(\cW,i)}_{k,0} &\land \bigwedge_{\ell = 0}^{k-1}\{\,\langle G \rangle_{= \ell}^{\text{pre}} \tau \mid \tau \in T_{k,d} \text{ and } \cW,i \models \langle G \rangle_{= \ell}^{\text{pre}} \tau\,\} \\
            &\land \{\,\langle G \rangle_{\geq k}^{\text{pre}} \tau \mid \tau \in T_{k,d} \text{ and } \cW,i \models \langle G \rangle_{\geq k}^{\text{pre}} \tau\,\} \\
            &\land \bigwedge_{\ell = 0}^{k-1}\{\,\langle P \rangle_{= \ell}^{\text{suf}} \tau \mid \tau \in T_{k,d} \text{ and } \cW,i \models \langle P \rangle_{= \ell}^{\text{suf}} \tau\,\} \\
            &\land \{\,\langle P \rangle_{\geq k}^{\text{suf}} \tau \mid \tau \in T_{k,d} \text{ and } \cW,i \models \langle P \rangle_{\geq k}^{\text{suf}} \tau\,\}.
        \end{aligned}
    \]
    To ensure the uniqueness of each type, we use some canonical bracketing for them.
\end{definition}

\begin{lemma}\label{lemma:type-equivalence-implies-logic-equivalence}
    If two pointed words satisfy the same $\GPTLF_{\geq}$-type of depth $d$ and width $k$, then they cannot be separated by any formula of $\GPTLF_{\geq}$ of depth $\leq d$ and width $\leq k$.
\end{lemma}
\begin{proof}
    We prove the claim by induction over $d$.
    Let $k$ be fixed for the entire proof, and let $d=0$. Pointed words $(\cW,i)$ and $(\cV,j)$ satisfying the same $\GPTLF_{\geq}$-type of depth $0$ means that they agree on all tokens, and thus on all Boolean combinations thereof. Therefore they cannot be separated by any formula of depth $\leq 0$, no matter its width.

    Suppose then that the claim holds for all $d'<d$ and that $(\cW,i)$ and $(\cV,j)$ satisfy the same $\GPTLF_{\geq}$-type of depth $d$ and width $k$. We prove that $(\cW,i)$ and $(\cV,j)$ are indistinguishable by structural induction. Since $(\cW,i)$ and $(\cV,j)$ satisfy also the same $\GPTLF_{\geq}$-types of depths $<d$ and width $k$, they are indistinguishable by any formula of depth $<d$ and width $k$. Thus it is sufficient to consider only the cases where $\langle G\rangle^{\text{pre}}_{\geq \ell}$ or $\langle P\rangle^{\text{suf}}_{\geq \ell}$ are applied for any $\ell \leq k$.

    Suppose $\varphi=\langle G\rangle^{\text{pre}}_{\geq \ell}\psi$. Since $\psi$ has depth $d-1$, then by the induction hypothesis, the truth value of $\psi$ at any vertex is determined by the truth value of the $\GPTLF_{\geq}$-type of depth $d-1$ and width $k$ at that vertex. Consider the set of all $(k,d-1)-$types $\tau$ that imply $\psi$:
    \[
        S_\psi:=\{\tau \in T_{k,d-1}\mid \tau\models\psi\}.
    \]
    We write $n^{\text{pre}}_\tau(\cW,i)$ for the number of vertices in the prefix 
    whose $(k,d-1)$-type is $\tau$. Note that $\cW,i\models\langle G\rangle ^{\text{pre}}_{\geq k}\psi$ if and only if $\sum_{\tau\in S_\psi}n^{\text{pre}}_\tau(\cW,i)\geq k$. For equivalence, it thus suffices to show that $\sum_{\tau\in S_\psi}n^{\text{pre}}_\tau(\cW,i)\geq k$ if and only if $\sum_{\tau\in S_\psi}n^{\text{pre}}_\tau(\cV,i)\geq k$.

    First, note that since the type $\tau^{(\cW,i)}_{k,d}$ records $\langle G\rangle^{\text{pre}}_{=\ell}$ for each $\ell\in\{0,\dots,k-1\}$ and $\langle G\rangle^{\text{pre}}_{\geq k}$, it determines $\min(n^{\text{pre}}_\tau(\cW,i),k)$ for each $\tau \in T_{k,d-1}$. Since $(\cV,j)$ also satisfies $\tau^{(\cW,i)}_{k,d}$, this is equal to $\min(n^{\text{pre}}_\tau(\cV,j),k)$. Consider now two cases:
    \begin{enumerate}
        \item If $n^{\text{pre}}_{\tau_0}(\cW,i)\geq k$ for some $\tau_0\in S_\psi$, then $\min(n^{\text{pre}}_{\tau_0}(\cW,i),k)=\min(n^{\text{pre}}_{\tau_0}(\cV,j),k)=k$, and hence both sums are at least $k$.
        \item If $n^{\text{pre}}_{\tau}(\cW,i)<k$ for all $\tau\in S_\psi$, then since $\min(n^{\text{pre}}_\tau(\cW,i),k)=\min(n^{\text{pre}}_\tau(\cV,j),k)$, we must have $n^{\text{pre}}_\tau(\cW,i)=n^{\text{pre}}_\tau(\cV,j)$, meaning the sums are equal.
    \end{enumerate}
    We have thus shown that $\mathcal{W},i\models\langle G\rangle^{\text{pre}}_{\geq\ell}\psi$ if and only if $\mathcal{V},j\models\langle G\rangle^{\text{pre}}_{\geq\ell}\psi$. The case $\varphi=\langle P\rangle^{\text{suf}}_{\geq\ell}$ is handled analogously.
\end{proof}

Next, we define type automata for the above defined types. These are, intuitively, $\CPG$-automata that have maximal ability to distinguish vertices, i.e., if two vertices have the same state after $t$ iterations of a type automaton, they also have the same state after $t$ iterations of any other $\CPG$-automaton (that have the same width bound $k$ for multisets as the type automaton). This notion is analogous to counting type automata in \cite{ahvonen_2024} and type automata in \cite{kuusisto2013}.

A \textbf{$\GPTLF_{\geq}$-type automaton of width $k$} is a $\CPG$-automaton $A = (Q, \pi, \delta, n, b)$ defined as follows. We let $Q$ be the set of all $\GPTLF_{\geq}$-types of width $k$ (up to depth $n$).\footnote{Here we abuse notation since formally $Q$ should be a set of bit strings. This is not a problem since the set of all $\GPTLF_{\geq}$-types of width $k$ up to depth $n$ is finite, and it can thus be mapped injectively to a set of bit strings of some length $m\in\mathbb{Z}_+$.} For the initializing function, we define 
\[
    \pi(p) := p \land \bigwedge_{p \neq p' \in \Sigma} \neg p'.
\]
For the transition function $\delta$, let $\tau$ be a $\GPTLF_{\geq}$-type of width $k$, and let $M$ and $M'$ be multisets of such types of the same depth as $\tau$ (when the depths of all types in $\{\tau\}, M$ and $M'$ do not match, we may define transitions arbitrarily). We define
\[
    \begin{aligned}
        \delta(\tau, M, M') := \tau_0 &\land \bigwedge_{\ell = 0}^{k-1} \{\, \langle G \rangle_{= \ell}^{\text{pre}} \sigma \mid M(\sigma) = \ell \,\} \\
        &\land \{\, \langle G \rangle_{\geq k}^{\text{pre}} \sigma \mid M(\sigma) = k \,\} \\
        &\land \bigwedge_{\ell = 0}^{k-1} \{\, \langle P \rangle_{= \ell}^{\text{suf}} \sigma \mid M'(\sigma) = \ell \,\} \\
        &\land \{\, \langle P \rangle_{\geq k}^{\text{suf}} \sigma \mid M'(\sigma) = k \,\},
    \end{aligned}
\]
where $\tau_0$ is the unique $\GPTLF_{\geq}$-type of depth $0$ that is consistent with $\tau$.

\begin{lemma}\label{Analogous_B6}
    Each formula $\varphi$ of $\GPTLF_{\geq}$ of modal depth $d$ and width $k$ has a logically equivalent finite disjunction of $\GPTLF_{\geq}$-types of depth $d$ and width $k$.
\end{lemma}
\begin{proof}
    Intuitively, we take the disjunction of all types of depth $d$ and width $k$ that imply $\varphi$ and show that it is logically equivalent to $\varphi$.

    Let $T_{k,d}$ denote the set of all $\GPTLF_{\geq}$-types of depth $d$ and width $k$. Let $\Phi:=\{\tau \in T_{k,d}\mid \tau \models \varphi\}$ and $\neg\Phi:=\{\tau \in T_{k,d} \mid \tau \models \neg \varphi\}$, and let $\bigvee\Phi$ denote the disjunction of the types in $\Phi$. Obviously, we have $\Phi\cap\neg\Phi=\emptyset$ and $\bigvee\Phi\models\varphi$. Note also that $\Phi$ is finite since $T_{k,d}$ is finite.

    To show that $\varphi\models\bigvee\Phi$, it suffices to show that $\Phi\cup\neg\Phi=T_{k,d}$. Suppose that $\tau \in T_{k,d}\setminus(\Phi\cup\neg\Phi)$. Then there exist pointed words $(\cW,i)$ and $(\cV,j)$ that satisfy $\tau$ such that $\cW,i\models\varphi$ and $\cV,j\models\neg\varphi$. Since $(\cW,i)$ and $(\cV,j)$ satisfy the same $\GPTLF_{\geq}$-type of depth $d$ and width $k$, then by Lemma \ref{lemma:type-equivalence-implies-logic-equivalence}, there can be no formula of $\GPTLF_{\geq}$ of modal depth $d$ and width $k$ that separates $(\cW,i)$ and $(\cV,j)$, but $\varphi$ is such a formula, which is a contradiction. Thus $\bigvee\Phi$ and $\varphi$ are logically equivalent.
\end{proof}
\begin{proposition}\label{proposition: types work}
        If two pointed words $(\cW,i)$ and $(\cV,j)$ satisfy exactly the same $\GPTLF_{\geq}$-type of width $k$ and modal depth $d+1$, then $(\cW,i)$ and $(\cV,j)$ satisfy exactly the same $\GPTLF_{\geq}$-type of width $k$ and modal depth $d$.
\end{proposition}
\begin{proof}
    The proof is a straightforward induction over $d$. It is easy to see from the definition of $\GPTLF_{\geq}$-types, that if $(\cW,i)$ and $(\cV,j)$ satisfy exactly the same $\GPTLF_{\geq}$-type of width $k$ and modal depth $1$, then $(\cW,i)$ and $(\cV,j)$ satisfy exactly the same $\GPTLF_{\geq}$-type of width $k$ and modal depth $0$. 
    
    Let $d$ be such that the statement of the proposition holds. If $(\cW,i)$ and $(\cV,j)$ satisfy exactly the same $\GPTLF_{\geq}$-type of width $k$ and modal depth $d+1$, then the different parts of their $\GPTLF_{\geq}$-type must be equal, i.e. the sets
        \[
        \tau^{(\cW,i)}_{k,0},
        \]
        \[
        \begin{aligned}
            \bigwedge_{\ell = 0}^{k-1}\{\,\langle G \rangle_{= \ell}^{\text{pre}} \tau \mid \tau \in T_{k,d} \text{ and } \cW,i \models \langle G \rangle_{= \ell}^{\text{pre}} \tau\,\},\\
            \{\,\langle G \rangle_{\geq k}^{\text{pre}} \tau \mid \tau \in T_{k,d} \text{ and } \cW,i \models \langle G \rangle_{\geq k}^{\text{pre}} \tau\,\},\\
            \bigwedge_{\ell = 0}^{k-1}\{\,\langle P \rangle_{= \ell}^{\text{suf}} \tau \mid \tau \in T_{k,d} \text{ and } \cW,i \models \langle P \rangle_{= \ell}^{\text{suf}} \tau\,\},\\
            \{\,\langle P \rangle_{\geq k}^{\text{suf}} \tau \mid \tau \in T_{k,d} \text{ and } \cW,i \models \langle P \rangle_{\geq k}^{\text{suf}} \tau\,\}
        \end{aligned}
    \]
    must be equal to their counterpart in $(\cV,j)$ (when $(\cW,i)$ is replaced with $(\cV,j)$).
    By the induction hypothesis, the above observation also holds for $\tau\in T_{k,d-1}$. Thus $\tau^{(\cW,i)}_{k,d} =\tau^{(\cV,j)}_{k,d}$. 
\end{proof}

\begin{lemma}\label{lem:type-equivalence-implies-state-equivalence}
    Two pointed words $(\cW,i)$ and $(\cV,j)$ satisfy exactly the same $\GPTLF_{\geq}$-type of width $k$ and modal depth $d$ if and only if they share the same state in each round up to $d$ for each $\CPG$-automaton.
\end{lemma}
\begin{proof}
    We will proceed by induction over the depth $d$. From left to right, we observe that the inputs to the automaton correspond in a natural way to the definition of types. From right to left, we observe that the result must hold for type automata where the transition function is injective. 
    
    Let us fix $k$, let $d=0$, and let $(\cW,i)$ and $(\cV,j)$ be pointed words. $(\cW,i)$ and $(\cV,j)$ satisfy the same $\GPTLF_{\geq}$-type of width $k$ and modal depth $0$ if and only if $\tau^{(\cW,i)}_{k,0}=\tau^{(\cV,j)}_{k,0}$. From the definition of $\CPG$-automata, $\tau^{(\cW,i)}_{k,0}=\tau^{(\cV,j)}_{k,0}$ holds if and only if each initialization function $\pi$ assigns $(\cW,i)$ and $(\cV,j)$ the same initial state.

    Let $d$ be such that the statement of the lemma holds. From the definition of types, it is easy to see that two pointed words $(\cW,i)$ and $(\cV,j)$ satisfy exactly the same $\GPTLF_{\geq}$-type of width $k$ and modal depth $d+1$ if and only if: 
    \begin{enumerate}
        \item They satisfy the same $\GPTLF_{\geq}$-type of width $k$ and modal depth $0$.
        \item For all $\GPTLF_{\geq}$-types $\tau$ of width $k$ and modal depth $d$, there are $l<k$ vertices in the prefix of $(\cW,i)$ that satisfy $\tau$ if and only if there are $l$ vertices in the prefix of $(\cV,j)$ that satisfy $\tau$.
        \item For all $\GPTLF_{\geq}$-types $\tau$ of width $k$ and modal depth $d$, there are $l\ge k$ vertices in the prefix of $(\cW,i)$ that satisfy $\tau$ if and only if there are at least $k$ vertices in the prefix of $(\cV,j)$ that satisfy $\tau$.
        \item For all $\GPTLF_{\geq}$-types $\tau$ of width $k$ and modal depth $d$, there are $l<k$ vertices in the suffix of $(\cW,i)$ before $i$ that satisfy $\tau$ if and only if there are $l$ vertices in the suffix of $(\cV,j)$ before $j$ that satisfy $\tau$.
        \item For all $\GPTLF_{\geq}$-types $\tau$ of width $k$ and modal depth $d$, there are $l\ge k$ vertices in the suffix of $(\cW,i)$ before $i$ that satisfy $\tau$ if and only if there are at least $k$ vertices in the suffix of $(\cV,j)$ before $j$ that satisfy $\tau$.
    \end{enumerate}
    Item 1. holds if and only if each initialization function $\pi$ assigns $(\cW,i)$ and $(\cV,j)$ the same initial state. Notice that by Proposition~\ref{proposition: types work}, $\tau^{(\cW,i)}_{k,d} =\tau^{(\cV,j)}_{k,d}$.
    
    Item 2. is equivalent to, there are $l<k$ vertices in the prefix of $(\cW,i)$ with the same $\GPTLF_{\geq}$-type, $\tau$, of width $k$ and modal depth $d$ if and only if there are $l$ vertices in the prefix of $(\cV,j)$ with the same $\GPTLF_{\geq}$-type, $\tau$, of width $k$ and modal depth $d$. 
    
    Item 3. is equivalent to, there are at least $k$ vertices in the prefix of $(\cW,i)$ with the same $\GPTLF_{\geq}$-type, $\tau$, of width $k$ and modal depth $d$ if and only if there are at least $k$ vertices in the prefix of $(\cV,j)$ with the same $\GPTLF_{\geq}$-type, $\tau$, of width $k$ and modal depth $d$. 
    
    Item 4. is equivalent to, there are $l<k$ vertices in the suffix of $(\cW,i)$ before $i$ with the same $\GPTLF_{\geq}$-type, $\tau$, of width $k$ and modal depth $d$ if and only if there are $l$ vertices in the suffix of $(\cV,j)$ before $j$ with the same $\GPTLF_{\geq}$-type, $\tau$, of width $k$ and modal depth $d$. 
    
    Item 5. is equivalent to, there are at least $k$ vertices in the suffix of $(\cW,i)$ before $i$ with the same $\GPTLF_{\geq}$-type, $\tau$, of width $k$ and modal depth $d$ if and only if there are at least $k$ vertices in the suffix of $(\cV,j)$ before $j$ with the same $\GPTLF_{\geq}$-type, $\tau$, of width $k$ and modal depth $d$. 

    From the induction hypothesis, if two pointed words $(\cW,i')$ and $(\cV,j')$ satisfy exactly the same $\GPTLF_{\geq}$-type of width $k$ and modal depth $d$, then they share the same state in each round up to $d$ for each $\CPG$-automaton. Thus from items 1-5, $(\cW,i)$ and $(\cV,j)$ share the same state in each round up to $d$ for each $\CPG$-automaton. Finally, from the transition function, $\delta$, for each $\CPG$-automaton, the state of $(\cW,i)$ at round $d+1$ is determined by the state at round $d$ and the respective multisets, $\mathcal{M}_k(Q)$, at round $d$. By the induction hypothesis and items 1-5, these multisets are the same for $(\cW,i)$ and $(\cV,j)$. Thus, they share the same state in round $d+1$ for each $\CPG$-automaton. 

    For the other direction, suppose that $(\cW,i)$ and $(\cV,j)$ are such that they share the same state in each round up to $d+1$ for each $\CPG$-automaton. In particular, this holds for each type automaton. From Proposition~\ref{proposition: types work}, the transition function of a type automaton is injective (i.e. if two pointed words $(\cW,i')$ and $(\cV,j')$ satisfy different $\GPTLF_{\geq}$-types of width $k$ and modal depth $d$, then $(\cW,i')$ and $(\cV,j')$ satisfy different $\GPTLF_{\geq}$-types of width $k$ and modal depth $d+1$). Thus, that $(\cW,i)$ and $(\cV,j)$ share the same state in each round up to $d+1$ in a type automaton, implies that they have the same $\GPTLF_{\geq}$-types of width $k$ and modal depth $d+1$.
\end{proof}

We are now at last ready to prove Theorem~\ref{theorem: automata to logic}.

\automatatologic*
\begin{proof}
    Intuitively, we construct a corresponding formula $\varphi_q$ for each state $q$ of the automaton such that $\varphi_q$ is true in a vertex if and only if the final state given to the vertex is $q$. Then we construct a tuple consisting of suitable disjunctions of such formulae.
    
    Let $A = (Q, \pi, \delta, n, b)$ be a $\CPG$-automaton with multiset bound $k$.
    For each $q \in Q$, define
    \[
        \Phi_q:=\{ \tau^{(\cW,i)}_{k,n}\in T_{k,d} \mid \text{ the state given by $A$ to } (\cW,i) \text{ in round }n \text{ is } q\}.
    \]
    We show that $\cW,i\models\bigvee_{\tau \in \Phi_q}\tau$ if and only if the state given by $A$ to $(\cW,i)$ is $q$. 
    If $\cW,i\models\bigvee_{\tau\in\Phi_q}\tau$, then $\cW,i\models\tau^{(\cV,j)}_{k,n}$ for some pointed $\Sigma$-labeled word $(\cV,j)$ that outputs $q$. This means that $(\cW,i)$ and $(\cV,j)$ satisfy the same 
    $\GPTLF_{\geq}$-type 
    of depth $n$ and width $k$. By Lemma \ref{lem:type-equivalence-implies-state-equivalence}, this means that $(\cW,i)$ and $(\cV,j)$ share the same state in $A$ in each round up to $n$. Since $A$ gives $(\cV,j)$ the state $q$ in round $n$, $A$ also gives $(\cW,i)$ the same state in round $n$. Conversely, if $A$ gives $(\cW,i)$ the state $q$ in round $n$, then $\tau^{(\cW,i)}_{k,n}\in\Phi_q$ and thus $\cW,i\models\bigvee_{\tau \in \Phi_q}\tau$.

    Now we have a formula $\varphi_q = \bigvee_{\tau \in \Phi_q}\tau$ that is true in $(\cW,i)$ exactly when $i$ is in the state $q$ in round $n$. For each $m \in [b]$, let $Q_m$ denote the set of all states in $Q$ where the $m$th bit is $1$. Now we can construct the formula $\psi_m = \bigvee_{q \in Q_m} \varphi_q$, which is clearly true in $(\cW,i)$ if and only if the $m$th output bit of $(\cW,i)$ is~$1$. The equivalent formula tuple we seek is thus $(\psi_1, \dots, \psi_b)$.
\end{proof}

We finish this section off by proving Theorem~\ref{theorem:main}.

\skeletontheorem*
\begin{proof}
    Theorems~\ref{theorem: logic to transformers}, \ref{theorem: transformers to automata} and \ref{theorem: automata to logic} provide three of the translations. Translations from encoder--decoder transformers to $\GPTLF$ and from $\CPG$-automata to transformers are obtained by applying the translations back-to-back. From $\GPTLF$ to $\CPG$-automata, we first translate a tuple of $\GPTLF$ into a transformer using Theorem~\ref{theorem: logic to transformers}. Then we modify the construction from the proof of Theorem~\ref{theorem: transformers to automata} by increasing the output round by one and having the transition function essentially map the substrings representing the floating-point numbers $0$ and $1$ to the bits $0$ and $1$.
\end{proof}

\section{Proofs for Theorem~\ref{theorem: language characterization}}\label{appendix: accepting}

In this section, we prove Theorem~\ref{theorem: language characterization}.

\fourwayresult*

We begin with an auxiliary result that has importance concerning the computability of some of the translations. An infinite sequence of tokens $\vec\sigma = \sigma_1 \sigma_2 \cdots$ is \textbf{ultimately periodic} if there are $p \geq 0$ and $n \geq 1$ such that $\sigma_{t+n} = \sigma_t$ for all $t > p$. For the least such values of $p$ and $n$, we call $\sigma_1 \cdots \sigma_p$ the \textbf{preperiod} and $\sigma_{p+1} \cdots \sigma_{p+n}$ the \textbf{period} of $\vec\sigma$.
A period is called \textbf{repetition-free} if it does not contain the same token twice. We say that a token-generating system is ultimately periodic if every output sequence it generates is either ultimately periodic or finite.

We first prove Proposition~\ref{proposition: ultimately periodic}.

\ultimatelyperiodic*
\begin{proof}
    By Theorem~\ref{theorem: generation-equivalence}, we can turn the transformer into a sequence $\vec\varphi$ of formulae of $\GPTLF$ that generates the same output sequences as the transformer with the same inputs. 
    If the generated output with an input is finite, then we are done, so suppose that it is infinite.
    Intuitively, the proof rests on the observation that a formula $\varphi$ of $\GPTLF$ does not separate two vertices $i$ and $j$ in the suffix of a two-sorted word $\cW = (V, E, \lambda)$ if \textbf{(1)} $i$ and $j$ share the same token, i.e., $\lambda(i) = \lambda(j)$, and \textbf{(2)} for each subformula of $\varphi$ of the type $\langle P \rangle^{\text{suf}} \psi$, we have $\cW, i \models \langle P \rangle^{\text{suf}} \psi$ if and only if $\cW, j \models \langle P \rangle^{\text{suf}} \psi$.

    In more detail, let $\Phi$ be the (finite) set of subformulae of $\vec\varphi$, let $i_t$ denote the last suffix vertex at generation step $t$, and define
    \[
        H_t:=\{\psi \in \Phi \mid \mathcal{W},j\models\psi \text{ for some suffix vertex } j< i_t\}.
    \]

    Now, the truth value of each $\psi\in\Phi$ at $i_t$ at different generation steps $t$ depends only on the pair $(\lambda(i_t),H_t)$ because the truth value of each subformula of the type $\langle G\rangle^{\text{pre}}_{\geq k}\psi$ remains constant throughout the generation. Since the transformer is deterministic, both the next token $\lambda(i_{t+1})$ as well as the set $H_{t+1}=H_t\cup\{\psi\in\Phi \mid \mathcal{W},i_t\models\psi\}$ also depend only on the pair $(\lambda(i_t),H_t)$. The sequence $(\lambda(i_t),H_t)_{t\geq 1}$ is thus obtained by iterating a
    function of the type $(\Sigma\cup\{\BOS\})\times\mathcal{P}(\Phi) \to (\Sigma\cup\{\BOS\})\times\mathcal{P}(\Phi)$, meaning it is ultimately periodic. Finally, since $(H_t)_{t\geq 1}$ is non-decreasing, 
    $H_t$ must be constant in the period of $(\lambda(i_t),H_t)_{t\geq 1}$. So, in the period, the only remaining varying state is the current token $\lambda(i_t)$, meaning it can only appear once per period of $(\lambda(i_t),H_t)_{t\geq 1}$, and thus also in the period of $(\lambda(i_t))_{t\geq 1}$. 
\end{proof}

Next, we give the translation from sentences of CPL to accepting encoder--decoder transformers.

\begin{lemma}\label{lemma: counting propositional logic to accepting transformers}
    If a language is recognized by a sentence $\varphi$ of counting propositional logic, then 
    it is also recognized by a computable accepting encoder-decoder transformer.
\end{lemma}
\begin{proof}
    Intuitively, we modify and extend the transformer constructed in the proof of Theorem~\ref{theorem: logic to transformers} to be autoregressive.

    First, we translate $\varphi$ into a transformer without the final softmax using Theorem \ref{theorem: logic to transformers}. In a single additional decoder layer we use another $\MLP$ to calculate the truth value of $\neg\varphi$ as done for other subformulae in the construction in the proof of Theorem~\ref{theorem: logic to transformers}.
    According to the construction, the final linear transformation maps to a float vector of dimension $1$ that has $1$ if $\varphi$ is true and $0$ otherwise; we modify the linear transformation to map to float vectors of dimension $|\Sigma \cup \{\EOS\}|$ such that the truth value of $\varphi$ is simply mapped to the accepting token and the truth value of $\neg\varphi$ is mapped to some non-accepting token, and other elements of the vector are zero. 
    
    The final softmax thus outputs a probability distribution where the corresponding token has the highest probability, and defining the token-generating function $f:\{0,1\}^n\to\Sigma\cup\{\EOS\}$ such that it simply picks the token with the highest probability gives us a transformer that immediately outputs an accepting token if $\varphi$ accepts the input. Furthermore, since the transformer only speaks about the prefix input (because $\varphi$ is a sentence of CPL), it generates the same token each round, and hence it never generates an accepting token if $\varphi$ does not accept the input.

    This translation is computable since the translation in Theorem $\ref{theorem: logic to transformers}$ is computable and the modifications to it are obviously computable.
\end{proof}

Next, we translate accepting encoder--decoder transformers to bounded multiset functions.

\begin{lemma}\label{lemma: accepting transformers to multiset functions}
    If a language is recognized by an accepting encoder--decoder transformer, then it is also recognized by a computable $k$-capped multiset function 
    for some $k \in \N$.
\end{lemma}
\begin{proof}
    Intuitively, by Theorem~\ref{theorem: generation-equivalence} we can consider $\CPG$-automata instead of transformers and extract the result by looking at the state transition function.

    Let $T$ be an accepting transformer. By Theorem \ref{theorem: generation-equivalence}, we can turn $T$ into an equivalent autoregressive CPG-automaton $A$ with some multiset bound $k$. On the input side (i.e.\ the prefix), the equivalence classes induced by the tokens can never be refined further, meaning that if two input vertices share the same token, they also share the same state in every round. This is because the first input of the transition function is the token corresponding to the vertex and the other two inputs are shared across encoder vertices. Therefore, the maximal information ever obtained regarding the input is the $k$-capped multiset of tokens appearing in the input.
    This implies that any word generated by $A$ (and thus acceptance by $T$) is invariant under (i) permuting the input and (ii) inserting a copy of a token $\sigma$ that already occurs at least $k$ times in the input. This means the language recognized by $T$ can also be recognized by a $k$-capped multiset function.

    Finding out which $k$-capped multisets should map to $1$ is simple. We simply run the transformer on each input where each token appears at most $k$ times; for any input where a token appears more than $k$ times, the output sequence would be the same if the token only appeared $k$ times, so they do not need to be considered. By Proposition \ref{proposition: ultimately periodic}, the transformer only needs to run for a bounded number of iterations before its output sequence starts repeating itself, at which point we can terminate if an accepting token has not been generated. Thus, the translation above is computable.
\end{proof}

We close the loop by translating bounded multiset functions to sentences of CPL.

\begin{lemma}\label{lemma: multiset functions to counting propositional logic}
    If a language is recognized by a $k$-capped multiset function for some $k \in \N$, then it is also recognized by a computable sentence of counting propositional logic.
\end{lemma}
\begin{proof}
    We essentially construct a canonical formula based on the $k$-capped multisets that are mapped to $1$.

    Each $k$-capped multiset $M$ can be expressed as a conjunction of formulas $\langle G \rangle_{= \ell}^{\text{pre}} \sigma$ for each token $\sigma$ such that $M(\sigma) = \ell < k$ and $\langle G \rangle_{\geq k}^{\text{pre}}\sigma$ for each $\sigma$ such that $M(p) = k$. We take the disjunction of the formulae corresponding to $k$-capped multisets that the function maps to $1$. Clearly the resulting sentence is true in a word if and only if the corresponding multiset is accepted by the function.

    This translation is obviously computable.
\end{proof}

We complete the proof by showing that $k$-capped multiset functions recognize exactly the commutative star-free languages. A language $L$ over $\Sigma$ is \textbf{commutative} if it is closed under permutation, i.e., if $\vec\sigma\in L$ implies that every permutation of $\vec\sigma$ is in $L$.
For languages $L, L'$ over $\Sigma$, their \textbf{concatenation} is $L L' := \{\, \vec\sigma \vec\sigma' \mid \vec\sigma \in L,\ \vec\sigma' \in L' \,\}$ and the \textbf{complement} of $L$ is $\overline{L} := \Sigma^* \setminus L$. The class of \textbf{star-free} languages over $\Sigma$ is the smallest class of languages that contains $\emptyset$ and $\{\sigma\}$ for each $\sigma \in \Sigma$ and is closed under union, concatenation and complementation.

\begin{lemma}\label{lemma: multiset functions recognize commutative starfree languages}
    A language $L$ is recognized by a $k$-capped multiset function for some $k \in \N$ if and only if $L$ is commutative and star-free.
\end{lemma}
\begin{proof}
    From left to right, we intuitively construct a regular expression for the language $L$. For the converse we utilize Schützenberger's theorem \cite{Schutzenberger1965} which states that star-free languages are counter-free, which essentially means that, in terms of membership, the number of occurrences of a token only matters up to some bound.

    ($\Rightarrow$): Let $L$ be a language recognized by some $k$-capped multiset function $\mathcal{M}_k(\Sigma) \to \{0,1\}$. We construct a regular expression that defines $L$ as follows. First, given a multiset $M \in \mathcal{M}_k(\Sigma)$, consider the set $C$ consisting of those words that contain exactly $M(p)$ occurrences of each token $p$. Let $\mathcal{C}$ be the union of such sets $C$ across all multisets $M \in \mathcal{M}_k(\Sigma)$. We call each $w \in \mathcal{C}$ a canonical input, because the question of whether an input is accepted or not by the function can be reduced to asking whether a corresponding input in $\mathcal{C}$ is accepted.

    Next, let $w = \sigma_1 \cdots \sigma_n \in \mathcal{C}$ be a canonical input and let $S \subseteq \Sigma$ be the set of tokens $\sigma$ for which $\sigma$ is saturated in $w$, i.e., for which $\sigma$ appears $k$ times in $w$ (the maximum number of occurrences possible). Now $S$ is a language easily definable as the union of languages consisting of single tokens. We define the language $S^*$ consisting of those words that only contain saturated tokens via the regular expression
    \[
        S^* := \overline{\Sigma^*(\Sigma \setminus S) \Sigma^*},
    \]
    noting that $\Sigma^*$ is expressible without the Kleene star as $\overline{\emptyset}$.

    Now the language $\{w\}$ is definable as the concatenation of the languages $\{\sigma_1\}, \dots, \{\sigma_n\}$. We can insert $S^*$ into $n+1$ positions between any of the letters $\sigma_1, \dots, \sigma_n$, including at the very front before $\sigma_1$ and at the very back after $\sigma_n$. There are $2^{n+1}$ possible ways to insert $S^*$ into zero or more such positions, each yielding a language. We take the union of these languages, which contains exactly the words obtained from $w$ by adding saturated tokens. We obtain such a language for each canonical input $w$ whose corresponding multiset is classified $1$ by the function, and the union of these languages is exactly $L$, the language recognized by the function.

    ($\Leftarrow$): Let $L\subseteq \Sigma^*$ be commutative and star-free. We will show that there exists some $n \in \N$ such that for all words $w,w'\in\Sigma^*$ that share the same $n$-capped multiset of tokens, we have $w\in L \iff w'\in L$. This implies that $L$ is recognizable by a function $\mathcal{M}_n(\Sigma)\to\{0,1\}$.
    
    By Schützenberger's theorem \cite{Schutzenberger1965}, used in the form presented by McNaughton \& Papert \cite{mcnaughton1971counter}, since $L$ is star-free, it is \textbf{counter-free}, i.e., there is an $n \in \mathbb{N}$ such that
    \[
        x y^m z \in L \iff x y^n z \in L
    \]
    for all $x, y, z \in \Sigma^*$ and $m \geq n$. Now let $w$ and $w'$ be words with the same $n$-capped multiset. First, by commutativity, permute $w$ or $w'$ so that all identical tokens are next to each other and occur in the same order in both $w$ and $w'$. Then, a token $\sigma$ is either present exactly the same number of times in both words, so no adjustment is needed, or both counts are at least $n$, so the counter-free property of $L$ lets us add however many repetitions of the token to the word in which it appears less as is needed. Repeating this for every token $\sigma$ turns the words identical without affecting membership in $L$, so $w\in L \iff w'\in L$.
\end{proof}

We are now ready to prove Theorem~\ref{theorem: language characterization}.

\fourwayresult*
\begin{proof}
    Lemmas~\ref{lemma: counting propositional logic to accepting transformers}, \ref{lemma: accepting transformers to multiset functions} and \ref{lemma: multiset functions to counting propositional logic} together imply that all three classes recognize exactly the same languages. By Lemma~\ref{lemma: multiset functions recognize commutative starfree languages}, these are the commutative star-free languages.
\end{proof}

\section{Proofs for Theorem~\ref{theorem: transformer cpl generation equivalence}}\label{appendix: transformer cpl generation equivalence}

In this section, we prove Theorem \ref{theorem: transformer cpl generation equivalence}.

\matrixtheorem*

We start with the direction from transformers to logic.

\begin{lemma}\label{lemma: transformer to cpl}
    For each autoregressive encoder--decoder transformer, we can construct a generation-equivalent pair of matrices of sentences of CPL.
\end{lemma}
\begin{proof}
    We combine three facts: \textbf{(1)} a transformer can only distinguish a finite number of inputs, \textbf{(2)}~autoregressive transformers are ultimately periodic, and \textbf{(3)} a sentence of counting propositional logic can distinguish the same inputs as a transformer. Intuitively, \textbf{(1)} means that we only have to study a finite number of inputs, \textbf{(2)} guarantees that every output starts repeating after a bounded number of iterations, and \textbf{(3)} means that we can separate the tokens generated by the transformer in each position up to the point where they inevitably start repeating via sentences of CPL.

    Let $T$ be an autoregressive transformer and let $f:\{0,1\}^n\to \Sigma\cup\{\EOS\}$ be its token-generating function. Observe that there is a $k \in \mathbb{Z}_+$ such that the word generated by $T$ depends only on the $k$-capped multiset of tokens present in the input (see Lemma \ref{lemma: accepting transformers to multiset functions}). Since there are only finitely many such multisets, we may fix a finite set $I$ of inputs containing one representative per multiset. Let $\#_\sigma(\mathcal{W}_p)$ denote the number of $\sigma$-tokens in the input $\mathcal{W}_p$. For each $\mathcal{W}_p \in I$, define a sentence $\chi_{\mathcal{W}_p}$ of counting propositional logic stating that an input has the same $k$-capped multiset as $\mathcal{W}_p$ as follows:
    \[
        \chi_{\mathcal{W}_p} := \bigwedge_{\#_\sigma(\mathcal{W}_p) < k}
        \langle G\rangle^{\text{pre}}_{= \#_\sigma(\mathcal{W}_p)}\sigma
        \;\wedge\; \bigwedge_{\#_\sigma(\mathcal{W}_p) \geq k}
        \langle G\rangle^{\text{pre}}_{\geq k}\sigma.
    \]
    Since $\Sigma$ is finite, $\chi_{\mathcal{W}_p}$ is a finitary formula.

    For each $\mathcal{W}_p \in I$, let $\vec\sigma^{\mathcal{W}_p}$ be the token sequence generated by $T$ on the input $\mathcal{W}_p$, padded into an infinite sequence with $\mathrm{EOS}$-tokens if it is finite. By Proposition~\ref{proposition: ultimately periodic}, each $\vec\sigma^{\mathcal{W}_p}$ is ultimately periodic and the period is repetition-free; let $p_{\mathcal{W}_p}$ denote the length of its preperiod and $n_{\mathcal{W}_p}$ the length of its period, and set 
    $N := \max\{\, p_{\mathcal{W}_p} \mid \mathcal{W}_p \in I \,\}$ and $L := \max \{\, n_{\mathcal{W}_p} \mid \mathcal{W}_p \in I \,\}$.

    We are now ready to define our sentence matrices. For each word $\cW_p \in I$, let $\cW_s$ denote the suffix of length $N+L$ obtained by applying $T$ to $\cW_p$ autoregressively (padded with $\EOS$-tokens if the generation is finite). Let $T(\mathcal{W}_s(i)) \in \{0,1\}^n$ denote the bit-string output of the transformer at the $i$th node of the suffix.
    For $1 \leq r \leq n$, $1 \leq m \leq N$ and $1 \leq \ell \leq L$, let the $r$th sentence of the $m$th column of $\Psi_1$ and of the $\ell$th column of $\Psi_2$ be
    \[
        \bigvee \{\, \chi_{\cW_p} \mid \cW_p \in I,\
        \text{the $r$th bit of } T(\mathcal{W}_s(m)) \text{ is } 1 \,\}
    \]
    and
    \[
        \bigvee \{\, \chi_{\cW_p} \mid \cW_p \in I,\
        \text{the $r$th bit of } T(\mathcal{W}_s(N+\ell)) \text{ is } 1 \,\},
    \]
    respectively.
    Fix an input and let ${\mathcal{W}_p} \in I$ be its representative. Since it satisfies $\chi_{\mathcal{W}_p}$ and no other $\chi_{{\mathcal{W}_p}'}$, each
    column of the pair encodes the bit-string output of $T$ on the corresponding node of the suffix, and thus, via $f$, the token that $T$ generates: the preperiod columns encode the tokens of the first
    $N$ steps, and the period columns encode the tokens of steps $N+1$ through
    $N+L$.
    
    The pair thus generates the same tokens as $T$ on the first $N$ steps. Now consider what $T$ does after step $N$. Since $N$ is at least the
    preperiod of $\vec\sigma^{\mathcal{W}_p}$, $T$ is already inside its cycle, so from
    step $N+1$ onward it simply repeats its period of length $n_{\mathcal{W}_p} \leq L$,
    whose tokens are pairwise distinct. The pair, meanwhile, outputs the
    encoded tokens of steps $N+1, N+2, \dots$ in order until it is about to
    output a token that it has already output since the start of the
    period. The first $n_{\mathcal{W}_p}$ of these tokens are one full traversal of $T$'s
    cycle, hence all distinct, so no reset happens among them; the next
    token, whether it is encoded by the $(n_{\mathcal{W}_p}+1)$th period column or by the
    first one again if $n_{\mathcal{W}_p} = L$, is the first token of the cycle, which has
    already been output. So the pair resets exactly when $T$ starts its
    cycle over, and both go on repeating the same $n_{\mathcal{W}_p}$ tokens forever.
    Therefore the transformer and the sentence matrices are generation-equivalent.
\end{proof}

We then show the direction from logic to transformers.

\begin{lemma}\label{lemma: cpl to transformer}
    For each pair of matrices of sentences of CPL, we can construct a generation-equivalent autoregressive encoder--decoder transformer that uses the final $\softmax$.
\end{lemma}
\begin{proof}
    We intuitively construct the transformer such that it first simulates the columns of the first matrix from left to right, keeping track of the current column via the masked self-attention in decoder layers. Then, the transformer simulates the columns of the second matrix from left to right, keeping track of the current column via the last generated token; if a token repeats, the transformer returns to the beginning of the second matrix.

    Let $\Psi_1, \Psi_2$ be matrices of CPL-sentences with $n$ rows, let
    $N$ and $L$ be their numbers of columns, and let
    $f : \{0,1\}^n \to \Sigma \cup \{\EOS\}$ be the token-generating function. Fix an input $\cW$, let $\vec\sigma_1$ and $\vec\sigma_2$ be the token strings obtained from the binary matrices $\Psi_1(\cW)$ and $\Psi_2(\cW)$ via $f$, and let $\vec\sigma_2'$ be the longest prefix of $\vec\sigma_2$ that does not contain any token twice. Since the tokens $\sigma_1, \dots, \sigma_{k}$ of $\vec\sigma'_2$ are pairwise distinct, the successor function $\mathrm{succ}(\sigma_i) := \sigma_{i+1}$ for $i < {k}$
    and $\mathrm{succ}(\sigma_{k}) := \sigma_1$ is well-defined.

    We construct a transformer $T$ with output dimension $|\Sigma \cup \{\EOS\}|$ and a final $\softmax$, together with a token-generating function $f'$, such that $(T, f')$ is generation-equivalent to $(\Psi_1, \Psi_2)$ with $f$. The transformer computes, at every suffix vertex, three pieces of
    information. First, the truth values of all formulae occurring in $\Psi_1$ and $\Psi_2$, via cross-attention to the decoder as in
    Theorem~\ref{theorem: logic to transformers}. 
    Second, the current token, via the embedding.
    Third, the position $t$ of the vertex in the suffix, capped at $N+1$, via
    $N+1$ decoder layers that together check the number of preceding suffix vertices, i.e., they compute the truth values of the formulas $\langle P \rangle^{\text{suf}} \top$, $\langle P \rangle^{\text{suf}} \langle P \rangle^{\text{suf}} \top$, etc.
    
    Since the domain of $f$ is finite, an $\MLP$ can then use this information to implement the following look-up table to ensure that the correct token
    $\sigma$ is generated at position $t$: if $t \leq N$, then $\sigma$ is the token encoded by the $t$th column of $\Psi_1(\mathcal{W})$; if
    $t = N+1$, then $\sigma$ is the token encoded by the first column of $\Psi_2(\mathcal{W})$; if $t \geq N+2$, then
    $\sigma := \mathrm{succ}(\sigma')$ where $\sigma'$ is the current token.
    Finally, the linear transformation places a
    logit $C$ on the coordinate of $\sigma$ and $0$ elsewhere, where $C$ is a
    float large enough that $\exp_{\cF}(-C) = 0$; after $\softmax$, the
    coordinate of $\sigma$ has probability $1$ and all others exactly $0$.
    The token-generating function $f'$ can then be defined to read a $\softmax$ output such that it
    returns the unique token with non-zero probability, so $T$ generates
    $\sigma$. So at each step, $(T,f')$ generates the same
    token as $(\Psi_1, \Psi_2)$ with $f$, meaning they are generation-equivalent.
\end{proof}

We are now able to prove Theorem~\ref{theorem: transformer cpl generation equivalence}.

\matrixtheorem*
\begin{proof}
    The direction from transformers to logic is shown by
    Lemma~\ref{lemma: transformer to cpl} and the opposite direction is shown by Lemma~\ref{lemma: cpl to transformer}.
\end{proof}

\section{Variations of the transformer architecture}\label{appendix:transformer-variations}

In this section, we show that allowing a cross-attention layer to have multiple heads or adding layer normalization does not increase expressive power, and that we can obtain analogous results for non-strict masking with minor modifications to the proofs.

\subsection{Multi-head cross-attention}\label{appendix:multi-head-cross-attention}

Many papers note in passing that a multi-head self-attention layer with $h$ heads can be simulated with $h$ sequential single-head self-attention layers \cite{merrill2024, li2024chain}. Here, we prove the same for cross-attention.

\begin{theorem}
    A multi-head cross-attention layer with $h$ heads can be simulated by $h$ sequential single-head cross-attention layers.
\end{theorem}
\begin{proof}
We intuitively have a dedicated space in the internal dimension for each attention head, and the heads are computed one-by-one instead of all at once.

Let $X \in \cF^{n_x \times d}$ (the encoder output) and $Y \in \cF^{n_y \times d}$ (the decoder input) be the inputs to a multi-head cross-attention layer with $h$ heads. Recall that
\[
    \mathrm{MHCA}(X, Y) = \mathrm{concat}(H^{(1)}(X, Y), \dots, H^{(h)}(X, Y)) W^O,
\]
where each head $H^{(\ell)}(X, Y) \in \cF^{n_y \times d_v}$ is a single-head cross-attention computation with its own parameter matrices $W^Q_\ell, W^K_\ell, W^V_\ell$, and $W^O \in \cF^{h d_v \times d}$ is the output projection.

We simulate this with $h$ sequential single-head cross-attention layers. The key idea is to use $h$ designated "slots" of width $d_v$ in the hidden dimension of $Y$ to accumulate the outputs of each head one at a time. We reserve dimensions $(\ell - 1) d_v + 1, \dots, \ell d_v$ of $Y$ for the $\ell$-th head's output, assuming $d \geq h d_v$ (otherwise we widen $Y$ and $X$ by padding with zeros, which does not affect expressivity).

At layer $\ell$, we apply a single-head cross-attention layer with parameters $(W^Q_\ell, W^K_\ell, W^V_\ell)$, producing $H^{(\ell)}(X, Y) \in \cF^{n_y \times d_v}$. We then use the residual connection and an MLP to write this output into the $\ell$-th slot of $Y$, leaving the other slots unchanged. Specifically, the MLP can be configured to zero out any interference from the residual connection in the target slot and to preserve the values in all other slots.

After $h$ layers, the hidden state contains $H^{(1)}(X, Y), \dots, H^{(h)}(X, Y)$ concatenated in the reserved slots, exactly as in the output of $\mathrm{concat}(H^{(1)}(X, Y), \dots, H^{(h)}(X, Y))$. The final linear projection $W^O$ can be implemented by the MLP of the last layer.

Since the encoder output $X$ is a fixed side input and is not modified by any of the $h$ layers, each single-head cross-attention layer sees the same $X$ as the original multi-head layer, and the computation within each head is identical to the corresponding head in the multi-head layer. Therefore the output of the $h$-layer simulation equals $\mathrm{MHCA}(X, Y)$.
\end{proof}

For the rest of the Appendix, we thus assume that each self- or cross-attention layer contains only a single attention head, which makes the proofs simpler.

\subsection{Masking}\label{appendix:masking}

It has been shown that transformers with strict causal masking are strictly more expressive than those with non-strict causal masking \cite{yang2024masked, licharacterizing}. We note that our characterizations easily cover both strict and non-strict causal masking with only slight modification to the logic and automaton class.

First, we introduce the alternative encoder--decoder architecture. It is otherwise the same as the one we have presented in Section~\ref{subsection:encoder-decoder-transformers}, except that the diagonal is not masked in the decoder layers, i.e., we use the masking function
\[
  (\mathrm{mask}(S))_{ij} = \begin{cases}
      -\infty &\text{when } j > i, \\
      S_{ij} &\text{otherwise}.
  \end{cases}
\]
The difference to the main section is that equality is not included in the upper condition.

Next, consider the logic $\GPTLF_2$ obtained from $\GPTLF$ by replacing the modality $\langle P \rangle^{\text{suf}}$ with the modality $\langle P \rangle_{\geq k}^{\text{suf}}$, with the following semantics. Given a pointed two-sorted word $(\cW, i)$, we have
\[
    \cW,i\models \langle P\rangle_{\geq k}^{\text{suf}}\,\varphi \text{ if and only if }  |\{j\in V(\cW_s) \mid j \leq i \text{ and }\ \cW,j\models \varphi\}| \geq k.
\]
Note the difference to the abbreviation defined in Appendix~\ref{appendix: automata to logic}, since here $\langle P\rangle_{\geq k}^{\text{suf}}$ can potentially count the vertex $i$ itself.

Thirdly, consider $\CPG'$-automata where the automata are otherwise defined the same as $\CPG$-automata in Section~\ref{subsection: automata} except that the second multiset of states is received from prior vertices in the suffix \emph{including the vertex itself}.

We obtain the following theorem.

\begin{theorem}
    Encoder--decoder transformers with non-strict causal masking, the logic $\GPTLF_2$ and $\CPG'$-automata have the same expressive power.
\end{theorem}
\begin{proof}
    The translations from transformers to automata and from automata to logic are analogous to the case with $\GPTLF$ and $\CPG$-automata. We only need to modify the translation from logic to transformers, where we intuitively show that the decoder can simulate the modality $\langle P \rangle_{\geq k}^{\text{suf}}$ similarly to how the encoder and decoder layers can simulate the modality $\langle G \rangle_{\geq k}^{\text{pre}}$. 
    
    We modify the proof of Theorem~\ref{theorem: logic to transformers} from Appendix~\ref{appendix: logic to transformers} by constructing a stack of decoder layers that simulates the modality $\langle P \rangle_{\geq k}^{\text{suf}}$ with the semantics defined above. The matrices $W_Q$ and $W_K$ are defined exactly as we did for the decoder layers that simulate $\langle G \rangle_{\geq k}^{\text{pre}}$. This means that before masking, each row contains the column of the input matrix that encodes the truth values of $\varphi$ multiplied by $F$, the largest positive non-infinite float in the format. After masking, entries to the right of the diagonal are set to $-\infty$, which are then set to $0$ by $\softmax$. For the other entries on each row, if at least one of them is $F$, then the bias and $\softmax$ turn each $0$ into $0$ and each $F$ into $\frac{1}{\ell'}$ (rounded to the nearest float) where $\ell'$ is the sum of $\ell$ $1$s and $\ell$ is the number of prior vertices (including self) that satisfy $\varphi$. This is the same outcome as in the proof of Theorem 47 of \cite{ahvonen2025expressive}, which means that we can construct the rest of the decoder in the same way, constructing the value matrix to trigger underflow when $\ell \geq k$. 
    
    Just like in \cite{ahvonen2025expressive}, there is the possibility of false positives; if all unmasked entries on a row are $0$, then biasing and $\softmax$ replaces each $0$ with $\frac{1}{m'}$ (rounded to the nearest float) where $m'$ is the sum of $m$ $1$s and $m$ is the number of prior vertices in the suffix including the vertex itself. Applying the value matrix construction then outputs $0$ for these rows, because each row of $XW_V$ corresponding to these non-zero entries is $0$. This is the same result as with underflow, so from the point of view of the vertex, the cases $\ell = 0$ and $\ell \geq k$ appear indistinguishable. To rectify this, we can simply use the masked self-attention head constructed in the proof of Theorem~\ref{theorem: logic to transformers} that simulates the non-counting modality $\langle P \rangle^{\text{suf}}$, which also works for the non-strict version of the modality and is enough to distinguish the cases $\ell = 0$ and $\ell \geq k$.

    The rest of the proof does not require modifications.
\end{proof}

We note that we can obtain results analogous to Theorems~\ref{theorem: softmax characterization}, \ref{theorem: generation-equivalence}, \ref{theorem: language characterization} and \ref{theorem: transformer cpl generation equivalence} and Proposition \ref{proposition: ultimately periodic} by modifying the proofs to consider non-strict masking instead of strict masking, $\GPTLF_2$ instead of $\GPTLF$ and $\CPG'$-automata instead of $\CPG$-automata. In particular, in the proof of Proposition~\ref{proposition: ultimately periodic} the set $H_t$ still saturates even with the non-strict past modality.

Alongside non-strict causal masking, \cite{vaswani2017attention} also shifts the decoder inputs to the right. However, this occurs outside the calculations of the transformer and thus does not affect our translations.

\subsection{Layer normalization}\label{appendix:layer-normalization}

It is known that layer normalization in a decoder self-attention layer does not affect expressive power \cite[Proposition B.13]{chiang2022overcoming, licharacterizing}. In this subsection, we show that the same holds for cross-attention, and that the choice between pre-norm and post-norm does not affect expressivity.

For each $\bx \in \cF^*$, let 
\[
    \mathrm{LN}(\textbf{x})=\frac{\textbf{x}-\mu}{\sqrt{\sigma^2+\varepsilon}}\cdot \gamma+\beta
\]
be a \textbf{layer normalization} function, where $\mu \in \cF$ and $\sigma \in \cF$ are the mean and standard deviation of the elements in $\textbf{x}$, $\gamma,\beta\in\mathcal{F}$ are learned parameters, and $\varepsilon\in\mathcal{F}$ is a small constant for numerical stability. 
A \textbf{decoder layer with residual connection and pre-norm} is obtained by replacing the multi-head cross-attention layer with residual connection ($\mathrm{MHCA}(X,Y) + Y$) with
\[
    \mathrm{MHCA}(\mathrm{LN}(X),\mathrm{LN}(Y))+Y,
\]
and one with \textbf{post-norm} is obtained by replacing it instead with
\[
    \mathrm{LN}(\mathrm{MHCA}(X,Y)+Y),
\]
where $\mathrm{LN}$ is applied row-wise.

\begin{theorem}
    A decoder layer with layer normalization (pre-norm or post-norm) can be simulated by a stack of decoder layers without layer normalization.
\end{theorem}
\begin{proof}
    When the dimension of vectors is fixed to $d$, $\mathrm{LN} : \cF^d \to \cF^d$ is a function on a finite domain and can therefore be implemented as a look-up table by a sufficiently wide two-layer ReLU MLP. For post-norm, $\mathrm{LN}$ is applied after $\mathrm{CA}$ and before the MLP sublayer, so 
    we can use the $\MLP$ to simulate $\mathrm{LN}$ and add a layer right after that carries out the other calculations of the $\MLP$. 
    For pre-norm, $\mathrm{LN}$ is applied to the inputs before $\mathrm{CA}$, so we 
    add a layer right before where the $\MLP$ simulates $\mathrm{LN}$.
\end{proof}

\begin{theorem}\label{thm:ln-does-not-decrease}
    A decoder layer without layer normalization can be simulated by a decoder
    layer with layer normalization (pre-norm or post-norm).
\end{theorem}
\begin{proof}
    We adapt the \emph{mirror dimension} trick of \cite{chiang2022overcoming} to
    cross-attention, meaning that we double the dimension and mirror the values on the left with inverse values on the right. Given inputs $X \in \cF^{n_x \times d}$ and
    $Y \in \cF^{n_y \times d}$, define mirrored versions
    $\tilde X \in \cF^{n_x \times 2d}$ and $\tilde Y \in \cF^{n_y \times 2d}$
    by setting $\tilde X_{i, d+j} = -X_{i,j}$ and similarly for $\tilde Y$.
    Every row of a mirrored matrix thus has mean $\mu = 0$.

    We show that cross-attention preserves this structure. Choose query and key projections
    $\tilde W^Q = \binom{W^Q}{0}$ and $\tilde W^K = \binom{W^K}{0}$, so that
    the attention scores depend only on the original halves and are identical to
    those of the unmirrored computation. Choose the value projection
    $\tilde W^V = \bigl(\begin{smallmatrix} W^V & -W^V \\ 0 & 0
    \end{smallmatrix}\bigr)$, so that $\tilde X \tilde W^V = (XW^V, -XW^V)$,
    which is mirrored. Since the attention weights multiply this from the left,
    the cross-attention output is mirrored, and its first $d$ columns equal the
    original cross-attention output. The residual $+\tilde Y$ preserves mirroring
    since mirrored matrices are closed under addition.

    Following \cite{chiang2022overcoming}, we control the row variance by reserving
    additional padding dimensions whose values are set (via the MLP sublayer) so
    that every row has the same norm. With $\mu = 0$ and $\sigma^2$ constant across
    all rows, choosing $\gamma = \sqrt{\sigma^2 + \varepsilon}$ and $\beta = 0$
    makes $\mathrm{LN}$ the identity. The argument applies to both pre-norm and
    post-norm, since in either case $\mathrm{LN}$ is applied to a mirrored matrix
    with controlled norms.
\end{proof}

\end{document}